%% file: main.tex
\documentclass{LMCS}

\def\dOi{8(3:29)2012}
\lmcsheading%
{\doi}
{1--35}
{}
{}
{Oct.~17, 2011}
{Sep.~30, 2012}
{}
 
\newcommand\mycitet[1]{\citeauthor{#1}~\cite{#1}}

\newcommand{\negation}[1]{#1}
\renewcommand{\negation}[1]{}

\usepackage{graphicx}
\graphicspath{{figures/}}

\usepackage{enumerate}
\usepackage{stmaryrd,textcomp}
\usepackage{ifxetex}
\usepackage{ifluatex}

\ifluatex
\usepackage{fontspec}
\else
\ifxetex
\usepackage{mathspec,xltxtra}
\else
\usepackage[utf8]{inputenc}
\fi
\fi

\usepackage[sort&compress,numbers]{dmnatbibfnsize} 
\usepackage{xcolor}
\usepackage{gawlitza}
\usepackage{ifdraft}

\usepackage{listings}
\lstdefinelanguage{prg}{otherkeywords={->,:=,|,;,>=,<=,:},keywords={->},keywordstyle=\textbf}
\usepackage{tikz}

\usepackage{hyperref}

\newcommand{\mytitle}{Invariant Generation 
  through 
  Strategy Iteration 
  in Succinctly Represented Control Flow Graphs}
\newcommand{\mykeywords}{%
  static program analysis,
  abstract interpretation,
  fixpoint equation systems,
  strategy improvement algorithms,
  SMT solving}

\hypersetup{pdftitle={\mytitle},
  pdfauthor={Thomas Gawlitza and David Monniaux},
  pdfkeywords={\mykeywords},
  pdfsubject={computation of least inductive invariants in polyhedral domains with fixed directions}}

\DeclareMathOperator{\defn}{{}\ensuremath{\mathrel{\mathop:}=}{}}
\newcommand{\complexclass}[1]{\textsf{#1}}
\newcommand{\piptwo}{$\mathsf{\Pi^p_2}$}
\newcommand{\sigmaptwo}{$\mathsf{\Sigma^p_2}$}

\let\origsection\section
\ifdraft{\renewcommand{\section}[1]{\newpage\origsection{#1}}}{}

\allowdisplaybreaks

\begin{document}


\title[Invariant Generation through Strategy Iteration]{\mytitle}

\author[T.~M.Gawlitza]{Thomas Martin Gawlitza\rsuper a}	
\address{{\lsuper a}School of Information Technologies, The University of Sydney, Australia}
\email{gawlitza@it.usyd.edu.au}
\thanks{{\lsuper a}This work was partially funded by the ANR project ``ASOPT''
  \raisebox{-1em}{\includegraphics[height=2em]{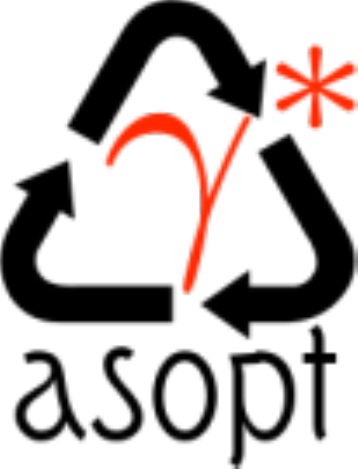}}}

\author[D.~Monniaux]{David Monniaux\rsuper b}	
\address{{\lsuper b}CNRS / VERIMAG Laboratory, Centre Équation, 2 avenue de Vignate, 38610 Gières, France}	
\email{David.Monniaux@imag.fr}  
\thanks{{\lsuper b}VERIMAG is a joint laboratory of CNRS, Université Joseph Fourier and Grenoble INP}



\keywords{\mykeywords}

\subjclass{D.2.4, F.3.1 (D.2.1, D.2.4, D.3.1, E.1), F.3.2 (D.3.1)}

\titlecomment{
  {\lsuper*}An earlier and shorter version of this article appeared in the proceedings 
  of the 21st European Symposium on Programming (ESOP) 2011~\cite{Gawlitza_Monniaux_ESOP11}.
}



\input{abstract}

\maketitle

\input{introduction}

\input{basics}

\input{model}

\input{eqs}

\input{max_strat_imp}

\input{improve}


\section{Complexity}
In this section, we shall prove that the decision problem associated with our computation is at the second level of the polynomial hierarchy, even if there is a single feedback vertex, a single real variable, and a single constraint in the template. It is therefore unsurprising that our algorithm exhibits exponential complexity in the worst case, by enumerating an exponential number of strategies: we shall then provide an artificial example on which it is the case.

\input{lower_bound}
\input{exp_example}
\input{upper}

\input{experiments}

\input{conclusion}

\section*{Acknowledgments}
The authors wish to thank the anonymous referees for their helpful suggestions and references.

\bibliographystyle{dmabbrvnat}
\bibliography{bib}

\end{document}

%% file: abstract.tex
\begin{abstract}
\noindent 
We consider the problem of computing numerical invariants of programs,
for instance bounds on the values of numerical program variables.
More specifically, we study the problem of performing static analysis
by abstract interpretation using template linear constraint domains.
Such invariants can be obtained by Kleene iterations that are, in
order to guarantee termination, accelerated by widening operators.
In many cases, however, applying this form of extrapolation leads to  invariants that are weaker than the strongest inductive invariant that can be expressed within the abstract domain in use.
Another well-known source of imprecision of traditional abstract interpretation techniques stems from their use of join operators at merge nodes in the control flow graph.
The mentioned weaknesses may prevent these methods from proving safety properties.

The technique we develop in this article addresses both of these issues:  contrary to Kleene iterations accelerated by widening operators, 
it is guaranteed to yield the strongest inductive invariant 
that can be expressed within the template linear constraint domain in use. 
It also eschews join operators by distinguishing all paths of loop-free code segments.
Formally speaking, our technique computes the least fixpoint 
within a given template linear constraint 
domain of a transition relation that is succinctly expressed as an existentially quantified linear real arithmetic formula. 

In contrast to previously published techniques that rely on quantifier elimination, our algorithm is proved to have optimal complexity:
we prove that the decision problem associated with our fixpoint problem is 
\piptwo-complete. 
Our procedure mimics a {\piptwo} search.
\end{abstract}

%% file: introduction.tex
\section{Introduction}
\label{s:introduction}

Static program analysis aims at deriving properties that are valid for all possible executions of a  program, through an algorithmic processing of its source or object  code. 
Examples of interesting properties include:
``the program always terminates'';
``the program never executes a division by zero'';
``the program never dereferences a null pointer'';
``the value of variable \lstinline|x| always lies between 1 and 3'';
``the output of the program is well-formed XHTML''.
There is considerable practical interest in being able to prove such properties automatically, in particular for software used in safety-critical applications, e.g., in fly-by-wire flight control systems in aircraft~\cite{DBLP:conf/safecomp/SouyrisD07}.

\subsection{Abstract interpretation}
It is well-known that fully automatic, sound and complete program analysis is impossible for any nontrivial property regarding the final output of a program.%
\footnote{This result, formally given within the framework of recursive function theory, is known as Rice's theorem~\cite[p.~34]{Rogers_theory_of_recursive_functions}\cite[corollary~B]{Rice_1953}.}
All analysis methods therefore suffer from at least one of the following limitations: they may be limited to programs with finite (and not too large) memory, or to bounded execution times;
they may be \emph{unsound}  (they may infer untrue properties);
or they may be \emph{incomplete} (they fail to prove certain true properties).
In this article, we use the \emph{abstract interpretation framework} of
\citet{CousotCousot_JLC92}
to construct a static analysis technique that is sound, but incomplete.

Static analysis by abstract interpretation replaces the computation over 
concrete \emph{reachable states} 
by computations over symbolically represented sets of concrete states.
The sets are taken from an \emph{abstract domain}.
For instance, one may aim at computing,
for each program point $p$ and each program variable $x$, 
an interval in which the value of $x$ is guaranteed to lie whenever the program reaches program point $p$. 
An analysis solely based on such intervals 
is known as \emph{interval analysis}~\cite{CouCou76}.
More refined numerical analyses include, 
for instance, 
finding for each program point an enclosing polyhedron for the vector of program variables~\cite{CousotHalbwachs78}. 
By restricting the analysis to handle only sets found within a particular abstract domain 
(e.g., Cartesian products of intervals or convex polyhedra), one can make the problem tractable, at the expense of \emph{over-approximation}.
For instance, if the domain in use consists of convex shapes, 
only, 
non-convex invariants will necessarily get over-approximated.

In addition to the abstract domain not being able to represent the required properties, 
a major source of imprecision 
is the use of \emph{widening operators} 
to enforce the convergence of Kleene iterations within finitely many iteration steps \cite{CousotCousot_JLC92}.
These operators extrapolate the first iterates of the Kleene sequence, say, of the intervals $[0,1]$, $[0,2]$, $[0,3]$, $\dots$ to a plausible limit, say $[0,+\infty)$, ensuring termination of the accelerated iteration. 
However,
such an accelerated iteration may overshoot the target, 
leading to further over-approximations
of the desired result.
In order to regain precision lost by widening,
one can then apply \emph{narrowing}. 
In its simplest form, narrowing is a descending iteration towards a fixpoint that strengthens 
the invariant step by step.
For more detailed information on Kleene iteration techniques in the context of abstract interpretation, 
we refer the reader to \citet{CousotCousot_JLC92}. 
Many variants of this basic iteration scheme have been proposed to alleviate the 
over-approximations introduced by widening 
\cite{DBLP:conf/sas/GopanR07,DBLP:conf/cav/GopanR06,Halbwachs_Henry_SAS2012}.
However, 
all these techniques do not guarantee to find the strongest inductive invariant that 
can be expressed in the abstract domain in use.

Let us illustrate the above mentioned weaknesses on the following simple example:

\noindent
\begin{minipage}{\textwidth}
\begin{lstlisting}[label=lst:loop_0_10_2]
i = 0;
while (true) {
  if (i < 10) i = i+2;
  else goto _end; }
_end: printf("i = %d\n");
\end{lstlisting}
\end{minipage}

\noindent
The strongest invariant, that is, the set of reachable states, is given by the proposition $i \in \{ 0, 2, 4, 6, 8, 10 \}$, which, 
together with the exit condition $i \geq 10$, yields $i = 10$ as the only possible final value of~$i$ at program point 
\lstinline$_end$.
Interval analysis by Kleene iterations with widenings computes the intervals $[0,0], [0,2], [0,4]$ 
and may then widen to $[0,+\infty)$.
The narrowing phase yields the inductive invariant $i \in [0, 11]$.
From this we can conclude that the final value of $i$ is in the interval $[10, 11]$.
The obtained interval $[0,11]$ represents the strongest \emph{inductive} 
invariant that can be expressed as an interval.%
\footnote{Some presentations of Hoare logic or static analysis call ``invariant'' what we refer to in this article as ``inductive invariant'': a set (or a logical formula defining such a set) containing all initial states and stable by the transition relation. In our terminology, an invariant is merely a property true at all times. With these definitions, an inductive invariant is an invariant by induction on the length of the execution trace, thus the terminology; however an invariant is not necessarily inductive.
Consider the initial state $(x,y)=(1,0)$ and a transition consisting in a {45\textdegree} clockwise rotation around $(0,0)$ : $(x,y) \in [-1,1]\times[-1,1]$ is an invariant (it is always true), but it is not inductive because $[-1,1]\times[-1,1]$ is not stable by this rotation.}
It is, however, not the strongest invariant expressible as an interval, 
which is $i \in [0,10]$. 
The invariant $i \in [0,10]$ is not inductive, 
because a state with $i = 9$ is mapped to a state with $i = 11$ by one iteration of the loop.

Unfortunately, small changes to the above program can make the widening/narrowing 
approach fail to produce a good invariant.
Consider, for instance, the introduction of an additional non-deterministic choice, represented by the function \lstinline|choice()|:

\noindent
\begin{minipage}{\textwidth}
\begin{lstlisting}[label=lst:loop_0_10_2_bypass]
i = 0;
while (true) {
  if (choice()) {
    if (i < 10) i = i+2;
    else goto _end; } }
_end: printf("i = %d\n");
\end{lstlisting}
\end{minipage}

\noindent
The program still outputs the value $10$, whenever it terminates. 
The only difference from the first version of the program is that there is, 
in each iteration, 
a non-deterministic choice whether or not the original loop body is to be executed.
If we perform the widening/narrowing technique on the modified version,
the widening phase will produce the same result $[0,+\infty)$.
However,  the narrowing phase is now not able to regain any precision lost due to widening.
The loop body represents the relation 
$\tau = \{ (i,i) \mid i \in \Z \} \cup \{ (i,i+2) \mid i \in \Z \text{ and } i < 10 \}$.
This relation is reflexive, that is, $(i,i)\in \tau$ for all $i \in \Z$.
The problem is of a general nature:
Whenever the transition relation $\tau$ of a loop is reflexive, 
descending iterations fail to improve the inductive invariant obtained by widening.

Of course, on such a simple example, one could use simple tricks to get rid of the imprecision and recover 
the interval $[0,11]$:
remove the identity from the transition relation (this does not change the set of all (inductive) invariants), 
or try a form of widening with thresholds, also known as widening ``up to''~\cite{Polka:FMSD:97}.
However, such approaches are brittle and may fail for more complex programs.

\subsection{Alternatives to the widening/narrowing approach}
\label{sec:closely_related_methods}
Because of the known weaknesses of the widening/narrowing approach, 
alternative methods have been proposed.
Finding an inductive invariant 
in an abstract domain can be recast as solving a constraint system.
Finding \emph{the strongest} inductive invariant 
is then the problem of finding a \emph{minimal} solution to the constraint system.
The technique described in this article is related to two recently proposed approaches, which we shall now briefly describe.

\subsubsection{Quantifier elimination}
\citet{Monniaux_LMCS10} considers abstract domains where elements are defined by a logical formula $I$ (more specifically, a conjunction of linear inequalities) that links the program variables to some parameters.
For instance, intervals on two variables $x,y$ are defined by $I \defn -l_x \leq x \leq u_x \land -l_y \leq y \leq u_y$,
 where $l_x,u_x,l_y,u_y$ are the parameters. An element from the abstract domain defined by the template $I$ 
 is specified by an assignment of values to the parameters.

Consider a set of initial states given by a formula $\iota$ 
(in the above example, with free variables $\sigma=(x,y)$) 
and a transition relation given by $\tau$ 
(in the above example, with free variables $(\sigma,\sigma')=(x,y,x',y')$). $I$ defines an inductive invariant for $\iota$ and $\tau$ if and only if
\begin{equation}
  \label{eqn:stability}
  \forall \sigma \,.\, \iota(\sigma) \Rightarrow I(\sigma) \land
  \forall \sigma,\sigma' \,.\, \left( I(\sigma) \land \tau(\sigma,\sigma') \Rightarrow I(\sigma') \right)
  .
\end{equation}

\noindent
Here, $I(\sigma)$ is the formula $I$ as 
above and $I(\sigma')$ is the formula $I$ with $\sigma$ replaced by~$\sigma'$.
The free variables of formula~\eqref{eqn:stability} are the parameters in~$I$.
In the above example, they are $l_x,u_x,l_y,u_y$.
Any satisfying assignment to these variables defines an inductive invariant
from the abstract domain.
A \emph{least} inductive invariant in the abstract domain is then defined by constructing, using formula~\eqref{eqn:stability} as
a building block, a formula whose solution is the minimal solution of~\eqref{eqn:stability}, 
using that, for any formula $F$, 
$x_0 = \min \{x \mid F(x)\}$ 
if and only if
\begin{equation}\label{eqn:quantified_stability}
F(x_0) \land \forall x \,.\, \left(F(x) \Rightarrow x_0 \leq x \right)
.
\end{equation}

\noindent
The static analyzer then proceeds as follows: transform the loop into 
a set of initial states 
$\iota$ 
and a transition relation $\tau$.
From these formulas, construct Formula~\ref{eqn:quantified_stability}. 
Then, call a solver capable of dealing with quantified formulas,
e.g, a quantifier elimination procedure or a lazy version thereof such as 
the one developed by \citet{Monniaux_CAV10}.

As an extension to this framework, $\iota$ and $\tau$ may have additional variables,
e.g., precondition or system parameters. 
The formula defining the least inductive invariant will then take the invariant parameters as a partial function (in the mathematical sense, that is, as a binary related each input to at most one output) of these precondition or system parameters. By quantifier elimination and further processing of the formula, it is possible to turn this formula 
into a closed-form function, and even into executable code computing that function 
(a tree of if-then-else statements with assignments at the leaves).

This approach allows to effectively synthesize best abstract transformers 
($\alpha \circ \tau \circ \gamma$ in the notation of \citet{CousotCousot_JLC92}). 
Unfortunately, quantifier elimination over linear real arithmetic is still very costly, despite the various recent works on this problem, and quantifier elimination over linear integer arithmetic and polynomial real arithmetic are even costlier.

The technique described in this article considers the same problem as the quantifier elimination approach, but without preconditions or system parameters. 
Our technique also uses a different algorithmic approach, called \emph{max-strategy iteration}.

\subsubsection{Strategy Iteration}
\label{sec:max_strategy}
In this article, we introduce a refinement of  the \emph{max-strategy iteration} technique of \citet{DBLP:conf/csl/GawlitzaS07}  for template linear constraint domains.
The phrase ``strategy iteration'', also known as ``policy iteration'', comes from game theory. Let us consider two-players zero-sum games: the outcome of such a game is a real number,
the two players (the maximizer and the minimizer) 
aim at maximizing (respectively, minimizing) the outcome.
Strategy iteration is a method for computing the optimal strategy for one of the players.
It successively improves a strategy through the following two steps until an optimal strategy is found:
(\emph{Evaluation}) 
Evaluate the currently selected strategy; and
(\emph{Improvement})
try to improve the currently selected strategy 
w.r.t.\ the result of the evaluation.

The max-strategy iteration technique of \citet{DBLP:conf/csl/GawlitzaS07}  
for finding invariants is inspired by this game-theoretic approach. 
Instantiated on template linear constraint domains, it computes the strongest inductive invariant  that can be represented by polyhedra 
of the form $P(b) = \{ x \in \R^n \mid Tx \leq b \}$,
where $T \in \R^{m\times n}$ is a template constraint matrix, which is fixed 
before the analysis is run (heuristics for finding a suitable matrix are out-of-scope for this article).
The variable $x$ is the vector of program variables.
The template constraint matrix $T$ is the counterpart 
of the template $I$ from the quantifier elimination technique of \citet{Monniaux_LMCS10}.
Given $T$, every vector $b \in \CR^m$ uniquely determines a polyhedron $P(b)$.
The vector $b$ contains the bounds on the linear functions that are represented by the rows of $T$.
With the appropriate choice of $T$
we can, among others, 
express the popular interval~\cite{CouCou76} 
and octagon~\cite{Mine_PhD,DBLP:journals/lisp/Mine06} abstract domains.

Similarly to Kleene iterations,
the max-strategy improvement algorithm produces an ascending sequence of pre-fixpoints 
that are less than or equal to the least inductive invariant we are aiming for.
The pre-fixpoints are obtained through convex optimization techniques, e.g., linear programming.
In contrast to Kleene iterations, though,
the algorithm converges to the least inductive invariant after at most exponentially many steps. 
Our conjecture is that it usually converges fast in practice, 
though one can concoct artificial examples that exhibit exponential behavior.

\subsubsection{Trace partitioning}
\label{sec:trace_partitioning}
Max-strategy iteration rids us of imprecisions introduced by widening, 
but, per se, does not remove imprecisions introduced by another operation: 
the merging of information from different program paths at join nodes in the control flow graph.
In this article, we introduce a refinement of max-strategy iteration 
where we distinguish the various execution paths, in a manner similar to the work of 
\citet{Monniaux_LMCS10}, 
and
\citet{Monniaux_Gonnord_SAS11}.

\begin{figure}
\begin{center}
\includegraphics[width=0.4\textwidth]{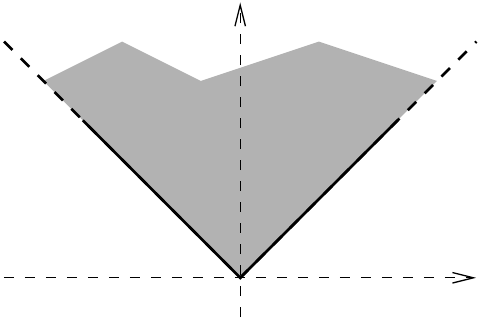}\hspace{0.1\textwidth}
\includegraphics[width=0.4\textwidth]{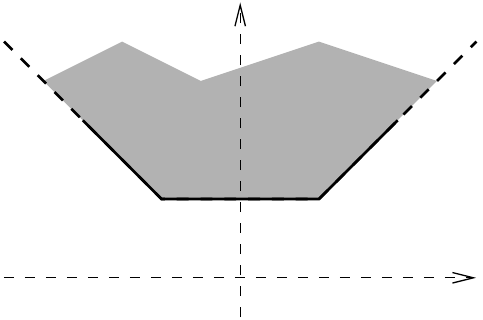}
\end{center}
\caption{On the left: the graph of $y=|x|$ is the union of two half-lines, but computing their convex hull yields the grayed shape. By intersection with $y \geq 1$, we obtain the shape on the right, which contains points with $x=0$ even though $y=|x| \land y \geq 1$ has no solution with $x=0$.}
\label{fig:x_xabs_ge1}
\end{figure}

In most systems for static analysis by abstract interpretation, joins in the control-flow graph result in computations of least upper bounds in the abstract domain. For instance, consider 
abstract interpretation over general convex polyhedra on the following program:

\noindent
\begin{minipage}{\textwidth}
\begin{lstlisting}
if (x >= 0) y = x; 
else y = -x;
if (y >= 1) z = 3.5/x;
\end{lstlisting}
\end{minipage}

\noindent
The program divides $3.5$ by 
the value of \lstinline|x| provided that the absolute value of \lstinline|x| is at least~$1$.
A static analyzer that uses convex polyhedra as abstract domain 
may work as follows.
%
After the first if-then-else statement, 
a convex hull is computed between the $x \geq 0 \land y = x$ and $x < 0 \land y = -x$ half-lines, 
resulting in a much larger polyhedron (see Fig.~\ref{fig:x_xabs_ge1}). 
The imprecision introduced by this operation prevents the analyzer from proving 
that a division by zero at line~$3$ is impossible.

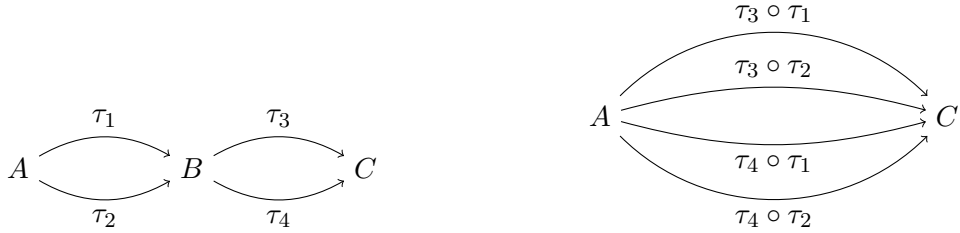
\begin{figure}
\begin{center}
\begin{tikzpicture}[node distance=6em]
\node (A) {$A$};
\node (B) [right of=A] {$B$};
\node (C) [right of=B] {$C$};
\path (A) edge[bend left,->] node[above] {$\tau_1$} (B);
\path (A) edge[bend right,->] node[below] {$\tau_2$} (B);
\path (B) edge[bend left,->] node[above] {$\tau_3$} (C);
\path (B) edge[bend right,->] node[below] {$\tau_4$} (C);
\end{tikzpicture}
\hspace{6em}
\begin{tikzpicture}[node distance=12em]
\node (A) {$A$};
\node (C) [right of=A] {$C$};
\path (A) edge[bend left=45,->] node[above] {$\tau_3 \circ \tau_1$} (C);
\path (A) edge[bend left=15,->] node[above] {$\tau_3 \circ \tau_2$} (C);
\path (A) edge[bend right=15,->] node[below] {$\tau_4 \circ \tau_1$} (C);
\path (A) edge[bend right=45,->] node[below] {$\tau_4 \circ \tau_2$} (C);
\end{tikzpicture}
\end{center}
\caption{Instead of considering two transitions (corresponding to a first if-then-else) followed by convex hull followed by two transitions (corresponding to a second if-then-else), as on the left, we get better precision by considering the four product transitions, as on the right.}
\label{fig:2x2trans_to_4trans}
\end{figure}

One solution is to get rid of all convex hulls corresponding to control flow joins by removing all control flow joins, except those corresponding to loop headers, by combining control flow edges. For instance, $n$ successive if-then-else constructs can be turned into an expanded system of $2^n$ transitions (Figure~\ref{fig:2x2trans_to_4trans} 
shows this construction for $n=2$). 
This is close to the \emph{trace partitioning} approach of \citet{Rival_Mauborgne_TOPLAS07}.%
\footnote{Trace partitioning analyses each program statement in different contexts according to an abstraction of the history of the control trace; thus, if a statement is preceded by $n$ tests, it can potentially analyze this statement in $2^n$ contexts. Because of this exponential blowup of maximal partitioning, trace partitioning techniques, including those implemented in Astr\'ee \cite{ASTREE:PLDI03,BlanchetCousotEtAl02-NJ}, use heuristics to ``fold'' abstract elements together using join operations.}
One could therefore run this exponential transformation first, and then run max-strategy iteration or min-strategy iteration (Sec.~\ref{sec:min_strategy}).
However, this transformation causes an exponential blowup and is  
therefore clearly not scalable.

In this article, we describe an algorithm that yields the same result as max-strategy iteration on this exponentially larger system. 
Our algorithm uses only polynomial space.
It achieves this by keeping the exponentially large system implicit.

\subsubsection{Path focusing}
\citet{Monniaux_Gonnord_SAS11,Henry_Monniaux_Moy_SAS2012} propose to run the classical Kleene iterations with widening and narrowing scheme not on the original control-flow graph, but on this exponentially larger system. In this approach, iterations are run on a distinguished subset of the original control nodes, such that all cycles in the original control flow graphs cross at least one of these distinguished nodes, using transitions corresponding to the simple paths between these distinguished nodes in the original control flow graph. The expanded control multigraph is kept implicit: the transitions, corresponding to simple paths in the original graph, are obtained on demand as solutions to SMT problems.
This approach has the following advantages:
\begin{enumerate}[(1)]
\item It fully does away with imprecisions introduced by ``join'' operations, except those corresponding to loops.
\item The transition relations on the simple paths may be accelerable.
          That is, they can be dealt with through acceleration techniques (cf.\ Sec.~\ref{sec:acceleration}, 
          \cite{DBLP:conf/sas/GonnordH06,Gonnord_PhD,LEROUX-SUTRE-SAS2007}).
\item While it uses widening operators, it does away with some of the imprecisions they introduce by focusing on one path at a time, which allows the use of narrowing iterations even on programs where they fail to yield better precision with the classical iteration scheme.
\end{enumerate}

\noindent
The technique we present in this article combines the idea 
of implicit representation with max-strategy iteration.

\subsection{Contributions}

The main contribution of this article is an algorithm that computes the 
strongest inductive invariant of the 
expanded transition system (which allows higher precision for abstract interpretation) 
without actually constructing it.
We shall see later the exact definition, but here is an interesting particular case (the general result allows more complex control flow):
given a $m \times n$ matrix $A$, an initial value $\iota \in \Q^n$ and a transition relation $\tau$ over $\Q^n$, defined by a formula over variables $x_1,\dots,x_n,x'_1,\dots,x'_n$, built with non-strict linear (in)equalities, $\land$, $\lor$ and prenex $\exists$, compute the least set of the form 
$P(b) = \{ x \in \R^n \mid Ax \leq b \}$ 
(that is, compute $b$) containing $\iota$ and stable by 
the transition relation $\tau$;
equivalently, find the least loop invariant of the form  $Ax \leq b$ for the loop with initial state $\iota$ and loop body expressed by~$\tau$. 

Our algorithm can be performed in polynomial space and exponential time.
It works in a demand-driven fashion: elements from the exponentially-sized sets of strategies and loop-free paths are enumerated only as needed, and one can thus hope that they will not all be enumerated, which seems to be confirmed 
by our preliminary experiments.

We also consider the following associated decision problem, which we shall later make more formal:
\begin{quote}
``Given a control-flow graph (with $N$ vertices) and transition relations written as existentially quantified 
first-order linear real arithmetic formulas, a family $A_1,\dots,A_N$ of matrices, an initial control state and a ``bad'' control state $b$, does there exist vectors $b_1,\dots,b_N$ such that $A_1 x \leq b_1 \land \dots \land A_N x \leq b_N$ forms an inductive invariant proving that $b$ is unreachable?''.
\end{quote}
We show this problem to be \sigmaptwo-complete (at the second level of the polynomial time hierarchy \cite[ch.~17]{Papadimitriou94}), even if $N=1$ and the matrix is $1 \times 1$.
Equivalently, the negated problem (abstract reachability of a statement) is shown to be \piptwo-complete.
Assuming the polynomial hierarchy does not collapse, 
this mean that this problem can be solved in polynomial space, but is harder than 
\complexclass{NP}-complete and \complexclass{coNP}-complete problems.
This clearly justifies the use of an exponential-time algorithm. 

\subsection{Other related Work}
\label{sec:min_strategy}\label{sec:acceleration}
Many approaches have been proposed to address the imprecisions caused by widening operators.
We now briefly describe approaches related to ours, in addition to those that we directly build upon 
(Sec.~\ref{sec:closely_related_methods}).
\citet{Polka:FMSD:97} proposed widening ``up to''
(an idea resurrected in the Astr\'ee system as widening with thresholds
\cite{BlanchetCousotEtAl02-NJ,ASTREE:PLDI03}),
which extracts syntactic hints for limiting widening.
\citet{DBLP:conf/sas/BagnaraHMZ05,DBLP:journals/scp/BagnaraHRZ05} proposed improvements over the ``classical'' widenings on linear constraint domains~\cite{Halbwachs_PhD}.
\citet{DBLP:conf/cav/GopanR06} introduced ``look-ahead widening'' \cite{DBLP:conf/cav/GopanR06} and ``guided iterations''~\cite{DBLP:conf/sas/GopanR07}: standard widening-based analysis is applied to a sequence of syntactic restrictions of the original program, which ultimately converges to the whole program; the idea is to distinguish phases or modes of operation in order to make the widening more precise.
Some other techniques fully do away with widenings~\cite{Colon_CAV03,Cousot05-VMCAI,Sankaranarayanan_SAS04}, for instance by expressing the invariants as solutions of a mathematical programming problem \cite{DBLP:journals/entcs/GoubaultRLLM10}, and thus the least invariant in the domain as an optimal solution to this problem.

In some cases, it is possible to compute exactly the transitive closure of the transition relation, or the application of the transitive closure to given initial states, or at least to compute a good over-approximation thereof. Such \emph{acceleration} techniques \cite{DBLP:conf/sas/GonnordH06,Gonnord_PhD,LEROUX-SUTRE-SAS2007} tend to have difficulties dealing with programs where the control flow is not flat (multiple paths within the loop body).

In Section~\ref{sec:max_strategy}, we sketched max-strategy iteration by an analogy to solving games where ``max'' operations correspond to control-flow joins and ``min'' operations to guards. If instead of choosing arguments to ``max'' operators, the strategy chooses them for ``min'' operators, we obtain min-strategy iteration \cite{DBLP:conf/cav/CostanGGMP05,DBLP:conf/esop/GaubertGTZ07}. 
Min-strategy iteration solves a sequence of fixpoint problems with decreasing values always 
weaker or equivalent to the strongest inductive invariant in the domain. 
In general, this sequence does not necessarily converge to this least inductive invariant, 
but it does so under certain conditions
(e.g., when all abstract transformers are non-expansive
\citep{arXiv:0806.1160}).
We investigated applying our ``implicit representation'' idea to the min-strategy approach, but encountered a stumbling block: while it is possible to decide whether a max-strategy is improvable
using SMT solving on quantifier-free formulas, 
the equivalent for min-strategies necessitated quantified formulas, which defeats the purpose of doing away with quantifier elimination techniques.


%% file: basics.tex
\section{Basics}
\label{s:basics}

\subsection{Notations}

$\mathbb{B} = \{ 0, 1 \}$ denotes the set of Boolean values.
The set of real numbers (resp.\ the set of rational numbers) 
is denoted by $\R$ (resp.\ $\Q$).
The complete linearly ordered set $\R \cup \{ \neginfty, \infty \}$
is denoted by $\CR$, similarly $\Q \cup \{ \neginfty, \infty \}$ is denoted
by $\CQ$.
%
For any expression (resp.\ term) $e$, 
we write
$
  e[e_1/\vx_1 , \ldots , e_k/\vx_k]
$
to denote the expression (resp.\  term) that is obtained from $e$ by simultaneously replacing 
all occurrences of the variables $\vx_1$, \ldots, $\vx_k$ by
$e_1,\ldots,e_k$.

A partially ordered set $\D$ is called a \emph{lattice} if and only if
any two elements $x,y \in \D$ have a greatest lower bound 
and a least upper bound,
denoted respectively by $x \wedge y$ and $x \vee y$.
It is a \emph{complete lattice} if and only if any subset
$X \subseteq \D$ has a greatest lower bound and a least upper bound,
denoted by $\bigwedge X$ and $\bigvee X$.
The least element $\bigvee \emptyset$
of a complete lattice is denoted by $\bot$.
The greatest element $\bigwedge \emptyset$
is denoted by~$\top$.

Assume that $\D_1$ and $\D_2$ are partially ordered by $\leq_1$ and $\leq_2$, respectively.
A function $f: \D_1 \rightarrow \D_2$ is called \emph{monotone} if and only if 
$f(x) \leq_2 f(y)$ for all $x, y \in \D_1$ with $x \leq_1 y$. 
We shall often use the following fundamental result:

\begin{thm}[Knaster/Tarski \cite{Tarski55}]
\label{th:tarski}
Let $\D$ be a complete lattice and $f: \D \rightarrow \D$ monotone. 
The operator $f$ has a least fixpoint and a greatest fixpoint,
respectively denoted by $\mu f$ and $\nu f$.
Moreover, 
we have
 $\mu f = \bigwedge \{ x \in \D \mid f(x) \leq x \}$ and 
 $\nu f = \bigvee \{ x \in \D \mid x \leq f(x) \}$.
 \qed
\end{thm}

\noindent
We denote the transpose of a matrix $A$ by $A^\top$.
For $x \in \CR$,
we denote the column vector $(x,\ldots,x)^\top$ by $\underline x$.
%
%
%
%
%
We denote the $i$-th row (resp.\ the $j$-th column) 
of a matrix $A$ by $A_{i\cdot}$ (resp.\ $A_{\cdot j}$).
Accordingly, 
$A_{i \cdot j}$ denotes the entry in the
$i$-th row and the $j$-th column.
We also use this notation for vectors and mappings $f : X \to Y^k$,
i.e., for all $i \in \{1,\cdots,k\}$,
the mapping $f_{i\cdot} : X \to Y$ is given by $f_{i\cdot}(x) = (f(x))_{i\cdot}$ 
for all $x \in X$.
The set $\CR^n$ is partially ordered by the component-wise extension of $\leq$,
which we again denote by $\leq$.
That is,
for all $x,y \in \CR^n$,
$x \leq y$ if and only if $x_{i\cdot} \leq y_{i\cdot}$ for all $i \in \{1,\ldots,n\}$.


A mapping $f : \CR^n \to \CR^m$ is called \emph{affine} 
if and only if there exist $A \in \R^{m \times n}$ and $b \in \CR^m$ such that 
$f(x) = Ax + b$ for all $x \in \CR^n$.
Here, we use the convention $\neginfty + \infty = \neginfty$.
Observe that $f$ is monotone 
if all entries of $A$ are non-negative.
A mapping $f : \CR^n\to\CR$ is called \emph{weak-affine} 
if and only if 
there exist $a \in \R^n$ and $b \in \CR$ such that 
$f(x) = a^\top x + b$ for all $x \in \CR^n$ with $f(x) \neq \neginfty$.
A mapping $f : \CR^n\to\CR^m$ is called \emph{weak-affine}
if and only if there exist weak-affine mappings 
$f_1,\ldots,f_m : \CR^n\to\CR$ such that $f = (f_1,\ldots,f_m)^\top$.
Every affine mapping is weak-affine, but not vice-versa.
In this article, we are concerned with mappings that are point-wise minimums of 
finitely many monotone and weak-affine mappings.
Note that these mappings are in particular concave, i.e., 
the set of points below the graph of the function is convex.

\subsection{Linear Programming}

Linear programming aims at optimizing a linear objective function with respect to linear constraints.
In this article,
we consider linear programming problems (LP problems for short)
of the form 
$ 
  \sup \,
  \{ 
  c^\top  x \mid 
  x \in \R^n
  ,
  Ax \leq b
  \}
$. 
  Here, 
  $A \in \R^{m \times n}$,
  $b \in \R^m$, and
  $c \in \R^n$
  are the inputs.
  The convex closed polyhedron
 $
  \{ 
  x \in \R^n
  \mid
  Ax \leq b
  \}
  $
  is called the \emph{feasible space}.
  The LP problem is called \emph{infeasible}
  if and only if the feasible space is empty.
  An element of the feasible space,
  is called \emph{feasible solution}.
  A feasible solution $x$ that maximizes $c^\top x$ is called
  \emph{optimal solution}.

If $A$ and $b$ consist of rational entries, only, 
then the feasible space is nonempty if and only if it contains a rational point.
An optimal solution exists if and only if there exists a rational one.
In this article, we always assume that all numbers in the input are rational.

LP problems can be solved in polynomial time
through the ellipsoid method~\cite{Kha:79} and interior point methods~\cite{LP1}. 
However, the running-time of these algorithms 
crucially depends on the sizes of occurring numbers.
At the danger of an exponential running-time in contrived cases, 
we can also instead rely on
the simplex algorithm: its worst-case running-time does not 
depend on the sizes of occurring numbers (given that arithmetic operations, comparison,
storage and retrieval for numbers are counted for $\mathcal{O}(1)$).
See for example \citet{LP1,Dantzig1998} for more information on linear programming.

\subsection{SAT modulo linear real arithmetic}

The set of SAT modulo linear real arithmetic formulas $\Phi$ 
is defined through the following grammar:
\begin{align}
  e      &::= c \mid x \mid e_1 + e_2 \mid c \cdot e'  &
  \Phi &::= a \mid e_1 \leq e_2 \mid e_1 < e_2 
    \mid \Phi_1 \vee \Phi_2 \mid \Phi_1 \wedge \Phi_2 \mid \neg \Phi'
\end{align}

\noindent
Here, $c \in \Q$ is a constant, $x$ is a real valued variable, $e, e',e_1,e_2$ are real-valued linear expressions, 
$a$ is a Boolean variable and $\Phi, \Phi',\Phi_1,\Phi_2$ are formulas.
An \emph{interpretation} $I$ for a formula $\Phi$ is a mapping that assigns a real value to every real-valued variable and
a Boolean value to every Boolean variable.
We write $I \models \Phi$ for ``$I$ is a \emph{model} of $\Phi$''.
That is,
we firstly inductively define a function $\sem{e}$ that evaluates 
a linear expression $e$ as follows:
\begin{align}
    \sem{c} I &= c 
    & 
    \sem{x} I &= I(x)
    &
    \sem{e_1 + e_2}I &= \sem{e_1}I + \sem{e_2}I
    &
    \sem{c \cdot e'}I &= c \cdot \sem{e'}I
\end{align}

\noindent
Secondly, 
we inductively define the relation $\models$ as follows:
\begin{align}
  \nonumber
    I \models a &\iff I(a) = 1
    &
    I \models e_1 \leq e_2 &\iff \sem{e_1}I \leq \sem{e_2}I
    \\[-1mm]
    I \models e_1 < e_2 &\iff \sem{e_1}I < \sem{e_2}I
    &
    I \models \Phi_1 \vee \Phi_2 &\iff I \models \Phi_1 \text{ or } I \models \Phi_2
    \\[-1mm]
    \nonumber
    I \models \Phi_1 \wedge \Phi_2 &\iff I \models \Phi_1 \text{ and } I \models \Phi_2
    &
    I \models \neg \Phi' &\iff I \not\models \Phi'
\end{align}

\noindent
A formula is called \emph{satisfiable} if and only if it has at least one model. 
A formula has a model if and only if it has a rational model.

The problem of deciding the satisfiability of SAT modulo linear real arithmetic formulas is \complexclass{NP}-complete.
There nevertheless exist efficient solver implementations for this decision problem, generally based on the 
DPLL(T) approach, an extension of the DPLL algorithm for SAT to richer logics.
For more information see 
for example \mycitet{Handbook_SAT},
\mycitet{DBLP:conf/cav/DutertreM06}, and
\mycitet{Kroening_Strichman_08}.
Such implementations, on satisfiable instances, can provide a model over Booleans and rational numbers.

In order to simplify notations we also allow matrices, vectors, 
the relations $\geq, \allowbreak >, \allowbreak \neq, \allowbreak =$, 
and the Boolean constants $0$ and $1$ 
to occur in SAT modulo linear real arithmetic formulas.


%% file: model.tex
\section{The Framework}
\label{s:model}

\subsection{Control Flow Graphs and Collecting Semantics}

In this article,
we model programs as control flow graphs,
i.e.,
a \emph{program} $G$ is a triple $(N,E,\start)$,
where 
\begin{enumerate}[(1)]
  \item
    $N$ 
    is a finite set of \emph{program points},
  \item
    $E \subseteq N \times \Stmt \times N$ 
    is a finite set of \emph{control-flow edges},
    and 
  \item
    $\start \in N$ is the \emph{start program point}.
\end{enumerate}

\noindent
A program uses $n$ real-valued variables $\vx_1,\ldots,\vx_n$.
A state is described by a vector $x \in \R^n$.
We assign a \emph{collecting semantics} 
$ 
  \sem{s} : 2^{\R^n} \to 2^{\R^n}
$ 
%
to each statement $s \in \Stmt$.
The collecting semantics $\sem{s}$ 
is an operator that
assigns a set $\sem{s}(X)$ of possible states after the execution of $s$
to a set $X$ of possible states before the execution of $s$.
The set $\Stmt$ of statements is specified subsequently.
The \emph{collecting semantics} $\Values$ of
a program $G = (N,E,\start)$ is finally defined as the least solution of 
the following constraint system:
\begin{align}
  \VALUES[\start]
    &\supseteq \R^n
  &
  \VALUES[v] 
    \supseteq \sem{s} (\VALUES[u])
  \quad \text{for all } (u,s,v) \in E
  .
\end{align}

\noindent
Here, 
for any $v \in N$,
the variable $\VALUES[v]$ takes values in $2^{\R^n}$.
The components of 
the collecting semantics $\Values$ 
are denoted by $\Values[v]$
for all
$v \in N$.
Throughout this article,
we will usually denote variables in bold face, 
and values in normal face.

\newcommand{\nonstrict}[1]{#1 [{<} / {\leq}]}

\subsection{Statements}
\label{sec:statements}
%
The set $\Stmt$ of all statements is the
set of all 
SAT modulo linear real arithmetic formulas
without Boolean variables and without negation.
Note that non-strict and strict inequality constraints are permitted.
The formula $e_1 \neq e_2$ is also permitted,
since it is an abbreviation for $e_1 < e_2 \vee e_2 < e_1$.
We can (in linear time) transform any 
SAT modulo linear real arithmetic formula without Boolean variables into this form by pushing negations to the leaves.

\negation{ 
  and \emph{without negation}.
  Note that this implies that we do not allow strict inequalities.
}%
The $\R$-valued variables 
$\vx_1,\ldots,\vx_n$ and $\vx_1',\ldots,\vx_n'$,
that may occur in the formula,
play a particular role.
The values of the variables 
$\vx_1,\ldots,\vx_n$ represent the values of the program variables
before executing the statement, 
and the values of the variables 
$\vx_1',\ldots,\vx_n'$ represent the values of the program variables
after executing the statement.
For convenience, 
we denote the vectors 
$(\vx_1,\ldots,\vx_n)^\top$ and $(\vx_1',\ldots,\vx_n')^\top$
also by $\vx$ and $\vx'$, respectively.
In addition to $\vx_1,\ldots,\vx_n$ and $\vx_1',\ldots,\vx_n'$, the
statement may also include other variables, which may stand for intermediate
values computed (or non-deterministically chosen) 
during the execution of a program statement. 
Conceptually, these variables are existentially quantified.

We could also add 
Boolean variables, at the expense of some additional complexity in definitions, theorems and proofs.
Note that this would not increase the expressiveness, 
since a Boolean variable $\vy$ 
can be simulated by a real variable $\widetilde{\vy}$ 
by replacing all occurrences of $\vy$ by $\widetilde{\vy} = 1$, 
all occurrences of $\neg \vy$ by $\widetilde{\vy}=0$, 
and conjoining $(\widetilde{\vy} = 0 \lor \widetilde{\vy} = 1$) 
to the formula.
In practice, 
the direct support of Boolean variables may be beneficial for the efficiency.
More generally, we can accommodate any formula feature that just expresses disjunctions in a compact way; the only requirement is not to generate negations.

The \emph{collecting semantics} $\sem{s} : 2^{\R^n} \to 2^{\R^n}$ of 
a statement $s \in \Stmt$ is defined by
\begin{align}
  \sem{s}(X)
  &:= 
  \{ 
    x' \in \R^n 
    \mid 
    \exists x \in X \,.\, 
    s [x / \vx, x' / \vx']
    \text{ is satisfiable}
  \}
  && 
  \text{for all } X \subseteq \R^n
  .
\end{align}

\noindent
Consider the following C-code snippet: 
\begin{lstlisting}
if (x_1 >= 0) 
  x_2 = x_1; 
else 
  x_2 = -x_1;
\end{lstlisting}

\noindent
Assume that \lstinline$x_1$ and \lstinline$x_2$ 
are of type \lstinline|int| and that they are the only numerical variables.
The effect of the C code snippet can be abstracted by the statement
\begin{align}
  \label{eq:stmt:001}
  \vx_1' 
  = 
  \vx_1 \wedge 
    \left(
      \left( \vx_1 \geq 0 \wedge \vx_2'  = \vx_1 \right)
      \vee
      \left( \vx_1 < 0 \wedge \vx_2'  = -\vx_1 \right)
    \right)
\end{align}


\noindent
Note that a conjunct $\vx'_i = \vx_i$ is needed for all variables that do not change their values.
\negation{
Note also that we have to write $\vx_1 \leq 0$ instead of $\vx_1 < 0$,
since we are not allowed to use strict inequalities;
since we do not allow negations, replacing strict inequalities by non-strict inequalities can add, not subtract transitions to the model, thus the resulting model is a sound abstraction of the original system.
If we additionally know that $\vx_1$ is an integer variable,
then we could also replace the strict inequality $\vx_1 < 0$ 
with the non-strict inequality $\vx_1 \leq -1$.
In the example, however, it does not matter at all.
}

A statement $s$ is called \emph{merge-simple}
if and only if it is in disjunctive normal form, i.e.,
$s$ is of the form $s_1 \vee \cdots \vee s_k$,
where the statements $s_1,\ldots,s_k$ do not use the Boolean connector~$\vee$.
Any statement can be 
rewritten into an equivalent 
merge-simple statement in exponential time and space
using distributivity.
The crux of our main result is that our algorithm never needs to compute
such an exponentially-sized disjunctive normal form.

If we convert Statement \eqref{eq:stmt:001}
into an equivalent merge-simple statement using distributivity,
we get:
\begin{align}
  \label{eq:stmt:002}
      \left( \vx_1' = \vx_1 \wedge \vx_1 \geq 0 \wedge \vx_2'  = \vx_1 \right)
      \vee
      \left( \vx_1' = \vx_1 \wedge \vx_1 < 0 \wedge \vx_2'  = -\vx_1 \right)
\end{align}

\noindent
A merge-simple statement $s$ that does not use the Boolean connector
$\vee$ at all 
is called \emph{sequential}.
Intuitively, sequential statements correspond to straight-line sequences of basic blocks.
The merge-simple statement \eqref{eq:stmt:002} non-deterministically 
chooses between executing one of the following 
sequential statement:
\begin{align}
      \vx_1' &= \vx_1 \wedge \vx_1 \geq 0 \wedge \vx_2'  = \vx_1
      &
%
%
      \vx_1' &= \vx_1 \wedge \vx_1 < 0 \wedge \vx_2'  = - \vx_1 
\end{align}

\subsection{Abstract Semantics}
\label{ss:abs:sem}

Let $\D$ be a complete lattice 
(for instance the complete lattice of all $n$-dimensional closed real intervals).
Assume that $\alpha : 2^{\R^n} \to \D$ 
and 
$\gamma : \D \to 2^{\R^n}$
form a Galois connection,
i.e.,
for all $X \subseteq \R^n$ and all $d \in \D$,
$\alpha(X) \leq d$ if and only if $X \leq \gamma(d)$.
The \emph{abstract semantics} $\sem{s}^\sharp : \D \to \D$ 
of a statement $s$
is then defined by 
\begin{align}
  \sem{s}^\sharp := \alpha \circ \sem{s} \circ \gamma
  .
\end{align}
Remark that we have chosen to use the best abstract transformer, i.e.,
the most precise abstract semantics.
All that was needed for soundness is that 
$\sem{s} \circ \gamma(d) \subseteq \gamma \circ \sem{s}^\sharp (d)$ for all $d \in \D$.
Our choice of $\sem{s}^\sharp(d)$, however, is the most accurate sound value.

The \emph{abstract semantics} $\Values^\sharp$ of a program
$G = (N,E,\start)$ is the least solution of the following 
constraint system:
\begin{align}
  \label{eq:sem:constraints}
  \VALUES^\sharp[\start]
    &\geq \alpha(\R^n)
  &
  \VALUES^\sharp[v] 
    \geq \sem{s}^\sharp (\VALUES^\sharp[u])
  \quad \text{for all } (u,s,v) \in E
\end{align}

\noindent
Here, 
for any $v \in N$,
the variable $\VALUES^\sharp[v]$ takes values in $\D$.
The components of 
the abstract semantics $\Values^\sharp$ 
are denoted by $\Values^\sharp[v]$
for all $v \in N$.
The abstraction is sound, i.e.,
the abstract semantics $\Values^\sharp$ safely over-approximates the collecting semantics $\Values$,
i.e., $\gamma(\Values^\sharp[v]) \supseteq \Values[v]$
for all 
$v \in N$.

\subsection{Template Linear Constraints}
\label{sec:template_domains}
In this article we restrict our considerations to \emph{template linear constraint domains}
as introduced by
\citet{DBLP:conf/vmcai/SankaranarayananSM05}.
We assume  
that a
\emph{template constraint matrix} $T \in \R^{m \times n}$ is given.
For technical convenience,
we always assume w.l.o.g.\ that $m \geq 1$ and each row of $T$ contains at least one non-zero entry.
The template linear constraint domain can be identified with the set $\CR^m$.
As shown by \citet{DBLP:conf/vmcai/SankaranarayananSM05},
the abstraction $\alpha : 2^{\R^n} \to \CR^m$ and
the concretization $\gamma : \CR^m \to 2^{\R^n}$,
which are defined by
\begin{align}
  \gamma(d) &:= \{ x \in \R^n \mid T x \leq d \} 
  && \text{for all } d \in \CR^m
  , \text{ and}
  \\
  \alpha(X) &:= \textstyle\bigwedge \{ d \in \D \mid \gamma(d) \supseteq X \}
  && \text{for all } X \subseteq \R^n
    ,
\end{align}
form a Galois connection.

The template linear constraint domains contain \emph{intervals},
\emph{zones}, and \emph{octagons} \cite{DBLP:journals/lisp/Mine06,Mine_PhD}, with appropriate choices of the
template constraint matrix $T$ 
\cite{DBLP:conf/vmcai/SankaranarayananSM05}.
For instance, if we have two variables $x$ and $y$, and we abstract each variable by an interval as $x \in [-l_x,u_x]$ and $y \in [-l_y,u_y]$, the vector $d$ is formed of $(l_x,l_y,u_x,u_y)$.
Here, the matrix $T$ is given by:
\begin{equation*}
T 
=
\begin{pmatrix}
-1 & 0\\
0  & -1\\
1  & 0\\
0  & 1\\
\end{pmatrix}
\end{equation*}

\noindent
and thus the concretization expresses:
\begin{equation*}
\gamma \begin{pmatrix}
l_x \\
l_y \\
u_x \\
u_y
\end{pmatrix}
=
\left\{ 
  \begin{pmatrix} x \\ y \\ \end{pmatrix} 
  \mid
  x \in [-l_x, u_x], \,
  y \in [-l_y, u_y]
\right\}
=
\left\{ 
  \begin{pmatrix} x \\ y \\ \end{pmatrix} 
  \mid
\begin{pmatrix}
-1 & 0\\
0  & -1\\
1  & 0\\
0  & 1\\
\end{pmatrix} 
\begin{pmatrix} x \\ y \\ \end{pmatrix} \leq
\begin{pmatrix}
l_x \\
l_y \\
u_x \\
u_y
\end{pmatrix}
\right\}
\end{equation*}

\noindent
While intervals, zones, and octagons are somewhat ``obvious'' choices, a common discussion with respect to template domains is how 
to find the templates, as opposed to the domain of convex polyhedra, where the convex hull and widening operations somewhat ``discover'' interesting directions in space. In this article, we shall assume that template matrices are given and refrain from discussing how they were obtained.

\newcommand{\davidaddon}[1]{}

\davidaddon{
Remark that, for a given $T$, the sets of states representable as $Tx \leq d$ form a lattice: the greatest lower bound (or intersection) of $\gamma(d_1)$ and $\gamma(d_2)$ is $\gamma(\inf(d_1,d_2))$, and their least upper bound is $\gamma(\sup(d_1,d_2))$, where $\inf(d_1,d_2)$ (resp. $\sup(d_1,d_2)$) is defined by coordinate-wise minimum (resp. maximum) of the vectors $d_1$ and $d_2$.
The negated lower bound in the interval above $[-l_x,u_x]$ is thus justified by consistency: the least upper bound (or convex hull) of two intervals $[-l_x,u_x]$ and $[-l_y,u_y]$ is $[-\max(l_x,l_y),\max(u_x,u_y)]$, the greatest lower bound (or intersection) is $[-\min(l_x,l_y),\min(u_x,u_y)]$; using $[a,b]$ would entail distinguishing upper and lower bounds in all constructs and computations, so as to maintain correct ordering.
}

\section{Improving the Precision of the Abstraction}

Most abstract interpretation techniques consider a control-flow graph with transitions expressed as sequential statements only (see formal definition in Sec.~\ref{sec:statements}), that is, composed of atomic guards and assignments.
An if-then-else construct with simple constructs (e.g., assignments) in both branches is thus expressed as two sequential statements, 
and a sequence of two such if-then-else constructs (one from point $A$ to point $B$ and one from $B$ to $C$) is expressed as on the left of Figure~\ref{fig:2x2trans_to_4trans}: 
two sequential statements between $A$ and $B$, 
and two sequential statements between $B$ and $C$.
As noted in the introduction (Sec.~\ref{sec:trace_partitioning}), 
abstract interpretation techniques usually abstract the set of reachable states at point $B$.
This may result in spurious states being considered in the abstraction, 
which in turn may result in the analysis tool being unable to prove desirable properties.

In this article, we apply an idea that is very similar to the \emph{path focusing} technique of \citet{Monniaux_Gonnord_SAS11}. Given a program expressed as a control-flow graph with sequential statements on the edges, we first compute a \emph{feedback vertex set} (a.k.a.\ \emph{cut-set}) $S$, that is, a set of control nodes 
(the \emph{feedback vertexes}) 
such that removing them cuts all cycles in the graph. Our original program is equivalent to a program 
where the only control nodes are those in the feedback vertex set, 
but edges carry arbitrary statements instead of sequential statements only 
(cf.\ Sec.~\ref{sec:statements}).
The results of program analyses on this new graph, at nodes from the
feedback vertex set $S$, are sound invariants for the original program. 
If information is needed at other nodes, 
we can compute it from the information we have for the nodes from $S$. 

Since methods for obtaining compact formulas expressing these statements from the original program have already been described in other publications 
\citep{Monniaux_Gonnord_SAS11}, 
we do not explain them in detail. Instead, we provide an example.

\tikzstyle{point}=[circle,draw,thick,inner sep=1pt,minimum size=2mm]
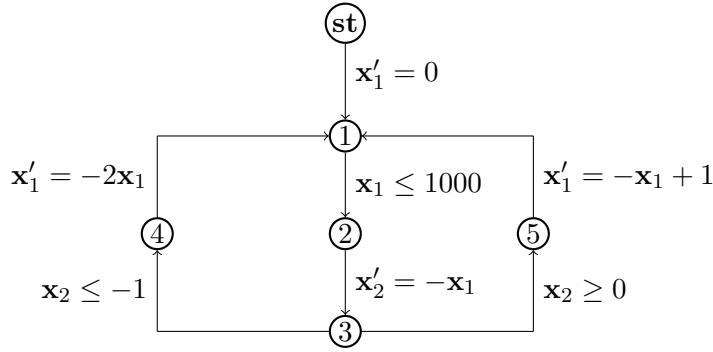
\begin{figure}
    \centering
	\scalebox{1}{
	\begin{tikzpicture}[]
		 \node (start) [point] {$\start$};
		 \node (n1) [below of = start,point,yshift=-5mm]{$1$};
		 \node (n2) [below of = n1,point,yshift=-3mm]{$2$};
		 \node (n3) [below of = n2,point,yshift=-3mm]{$3$};
		 \node (n4) [left of = n2,point,xshift=-15mm]{$4$};
		 \node (n4a) [coordinate,left of = n3,xshift=-15mm]{};
		 \node (n4b) [coordinate,left of = n1,xshift=-15mm]{};
		 \node (n5) [right of = n2,point,xshift= 15mm]{$5$};
		 \node (n5a) [coordinate,right of = n3,xshift= 15mm]{};
		 \node (n5b) [coordinate,right of = n1,xshift= 15mm]{};
		 \path[->] (start) edge [] node [right,yshift=1mm] {$\vx_1' = 0$} (n1);
		 \path[->] (n1) edge [] node [right] {$\vx_1 \leq 1000$} (n2);
		 \path[->] (n2) edge [] node [right] {$\vx_2' = - \vx_1$} (n3);
		 \path[-] (n3) edge [] node [below] {} (n4a);
		 \path[->] (n4a) edge [] node [left] {$\vx_2 \leq -1$} (n4);
		 \path[-] (n4) edge [] node [left] {$\vx_1' = -2 \vx_1$} (n4b);
		 \path[->] (n4b) edge [] node [above] {} (n1);
		 \path[-] (n3) edge [] node [below] {} (n5a);
		 \path[->] (n5a) edge [] node [right] {$\vx_2 \geq 0$} (n5);
		 \path[-] (n5) edge [] node [right] {$\vx_1' = -\vx_1 + 1$} (n5b);
		 \path[->] (n5b) edge [] node [above] {} (n1);
	\end{tikzpicture}
	}
	\caption{The program $G_1$ of the running example}
	\label{fig:run:ex:01:a}
	\bigskip
\end{figure}

\newcommand{\stmtrunningexample}{
		       \vx_1 \leq 1000 \wedge \vx_2' = -\vx_1\wedge
		       \left(
		         \left( \vx_2' \leq -1 \wedge \vx_1' = -2 \vx_1 \right)
		         \vee
		         \left( \vx_2' \geq 0 \wedge \vx_1' = - \vx_1 + 1 \right)
 		       \right)
}


\begin{exa}[Running Example]
  \label{ex:running:0}
  Throughout this article 
  we use the following C-code snippet
  as a running example:

\noindent
\begin{minipage}{\textwidth}
\begin{lstlisting}
int x_1, x_2; 
x_1 = 0; 
while (x_1 <= 1000) { 
    x_2 = -x_1; 
    if (x_2 < 0) x_1 = -2 * x_1; 
    else x_1 = -x_1 + 1;  }
\end{lstlisting}
\end{minipage}

\noindent
This C-code snippet is abstracted through the program 
$G_1 = (N_1,E_1,\start)$ depict in Figure~\ref{fig:run:ex:01:a}.
However, it is not necessary to apply abstraction at every program point,
i.e., to assign an abstract value to each program point.
It suffices to apply abstraction at a vertex feedback set of $G_1$.
Since all loops contain the program point $1$,
$\{ 1 \}$ is a feedback vertex set of $G_1$.
Equivalently to applying abstraction only at program point $1$, 
we can rewrite the control-flow graph $G_1$ 
into a control-flow graph $G = (N,E,\start)$ that is equivalent w.r.t.\ the collecting semantic,
but contains just the program point $\start$ 
and the program points from the vertex feedback set $\{1\}$.
The result of this transformation 
--- 
the control-flow graph $G$
---
is shown 
in Figure~\ref{fig:run:ex}(a) (Page \pageref{fig:run:ex}).
%
%

  The programs $G_1$ and $G$ are equivalent w.r.t.\ their collecting semantics, i.e.,
  $\Values[v] = \Values_1[v]$ for all $v \in N$.
  Here, 
  $\Values_1$ denotes the collecting semantics of $G_1$
  and $\Values$ denotes the collecting semantics of $G$.
  W.r.t.\ to the abstract semantics, $G$ is usually more precise than $G_1$,
  because we reduced the number of merge points.
  In general, we only have $\Values^\sharp[v] \subseteq \Values^\sharp_1[v]$ for all 
  $v \in N$,
  where $\Values^\sharp_1$ denotes the abstract semantics of $G_1$
  and $\Values^\sharp$ denotes the abstract semantics of $G$.
  This is independent of the abstract domain.\footnote{We assume that we have given a Galois-connection and thus in particular monotone best abstract transformers.}
\qed
\end{exa}

Let us make a few last remarks regarding the feedback vertex set.
Abstract interpretation techniques usually use such a set to select widening points~\citep[\S4.1.2]{Cousot_state_doctorate}.
In contrast, 
our method uses this set to select the nodes where it over-approximates the set of reachable states;
it does not over-approximate the set of reachable states at other nodes;
widening is not involved at all.
Finding a feedback vertex set of minimal cardinality is an NP-complete problem if the control-flow graph is arbitrary;
such a set can however be found in linear time 
if the control-flow graph is \emph{reducible} 
(in short, if loops have a single entry point)~\cite{DBLP:journals/siamcomp/Shamir79}, 
which is the case for control-flow graphs directly obtained from structured programs (the method extends to certain irreducible graphs).
The control-flow graph may however become irreducible if certain optimizations or partitioning techniques are used.
A common heuristic is, for structured programs, to use loop headers, 
and for unstructured programs to use the targets of back edges from a depth-first traversal~\cite{Bourdoncle93efficientchaotic,Bou92x};
this heuristic does not guarantee that the feedback vertex set is minimal with respect to inclusion ordering, let alone cardinality.

\section{Basic Observations}
\label{s:basic_obs}

We now note down basic properties of the abstract semantics.  

\subsection{Abstract Semantics of Statements}


Our first observation is that, 
for all \emph{sequential} statements $s$ and all $d \in \CR^m$,
$\sem{s}^\sharp (d)$ can be computed efficiently.

\begin{lem}[Sequential Statements]
\label{l:sequential:poly}
  Let $s$ be a sequential statement and $d \in \CR^m$.
  The operator $\sem{s}^\sharp$ is a 
  point-wise minimum of finitely many monotone and weak-affine operators.
  For all $d \in \CR^m$,
  $\sem{s}^\sharp(d)$ can be computed in polynomial time
  through linear programming.
\end{lem}

\begin{proof}
   Let $i \in \{ 1,\ldots,m \}$.
   We get:
   \begin{align}
     \label{eq:this:is:an:lp:1}
     \sem{s}^\sharp_{i\cdot}(d)
     &=
     \sup \, \left\{ T_{i\cdot} x' \mid x' \in \sem{s}(\gamma(d)) \right\}
     \\&=
     \label{eq:this:is:an:lp:2}
     \sup \, \left\{ 
       T_{i\cdot} x' \mid x' \in \R^n \text{ and } \exists x \,.\, T x \leq d \text{ and }  
       s[x/\vx,x'/\vx'] \text{ is satisfiable}
     \right\}
   \end{align}
   
   \noindent
   Equation~\ref{eq:this:is:an:lp:2} follows from Equation~\ref{eq:this:is:an:lp:1} by expansion of the concrete semantics $\sem{s}$ into a SMT-formula and of $\gamma(d)$ into $T x \leq d$.
  %
   Since $s$ does not contain disjunctions,
   the optimization problem in \eqref{eq:this:is:an:lp:2} aims at optimizing a linear objective function 
   w.r.t.\ linear constraints (equalities, strict inequalities, and non-strict inequalities).
   The optimal value of this
   optimization problem 
   can be computed in polynomial time through linear programming.
   To check feasibility by standard linear programming techniques 
   (which only allow non-strict inequalities),
   we can replace every strict inequality $e_1 < e_2$ 
   by the non-strict inequality $e_1 \leq e_2 - \epsilon$,
   where $\epsilon$ is appropriately small.
   Such an appropriately small $\epsilon$ can be computed in polynomial time.
   Provided that the optimization problem is feasible,
   we can then replace 
   $s[x/\vx,\,x'/\vx']$
   by
   $\nonstrict{s[x/\vx,\,x'/\vx']}$.
   Here, $\nonstrict{s}$ denotes the statement 
   obtained from $s$ by replacing every strict inequality relation by a non-strict inequality relation.
   The optimal value of the obtained linear programming 
   problem is equal to the optimal value of the 
   optimization problem \eqref{eq:this:is:an:lp:2}.
   
   %
   It remains to show that $\sem{s}^\sharp_{i\cdot}$
   is a point-wise minimum of finitely many monotone and weak-affine operators.
   Since 
   $\nonstrict{s[x/\vx, \, x'/\vx']}$
   is a conjunction of non-strict linear inequalities,
   there exist matrices $A$, $A'$ and $A''$ and a vector $b$ such that,
   for all $x$ and $x'$,
   $\nonstrict{s[x/\vx, \, x'/\vx']}$ is satisfiable
   if and only if there exists a $x''$ such that
   $A x + A' x' + A'' x'' \leq b$
   (the vector $x''$ stands for the other variables in~$s$,
   which are implicitly existentially quantified).
   Thus, the optimization problem 
   \eqref{eq:this:is:an:lp:2} 
   can be rewritten as follows:
   \begin{align}
     \sem{s}^\sharp_{i\cdot} (d)
     \!=\!
     \sup \left\{ 
       T_{i\cdot} x' 
       \!\mid\!
        x' {\in} \R^n, 
       \exists x {\in} \R^n \,.\,
       \exists x'' {\in} \R^q \,.\, 
       T x \leq d \text{ and } A x + A' x' + A'' x'' \leq b 
     \right\}
   \end{align}
   
   \noindent
   Strong duality \cite{Boyd_Vandenberghe_CVXOPT}, also known as Farkas' lemma,
   thus gives us, provided that $\sem{s}^\sharp_{i\cdot} (d) > \neginfty$,
   i.e., 
   the optimization problem is feasible,
   the following equation:
   \begin{align}
   \label{eqn:dual_lp}
     \sem{s}^\sharp_{i\cdot} (d)
     =
     \inf \, \left\{ 
       d^\top y_1 {+} b^\top y_2
       \mid 
       y_1,y_2 \geq 0 , \,
       T^\top y_1 {+} A^\top y_2 =0 , \,
       {A''}^\top y_2 = 0 ,\, 
       {A'}^\top y_2 = T_{i\cdot}^\top
     \right\}
   \end{align}
   
   \noindent
   Since $y_1 \geq 0$ for all feasible solutions of the linear programming problem
   in \eqref{eqn:dual_lp}, 
   $\sem{s}^\sharp_{i\cdot}$ 
   coincides with a point-wise infimum of monotone and affine operators on 
   the set $\{ d \in \CR^m \mid \sem{s}^\sharp_{i\cdot} (d) > \neginfty \}$.
   That is,
   $\sem{s}^\sharp_{i\cdot}$ is a point-wise infimum of monotone and weak-affine operators.
   Since the optimal value, provided that it exists, is attained at the vertices of the feasible space (finitely many),
   the point-wise infimum is a point-wise minimum of finitely many monotone and weak-affine operators.
\end{proof}

The max-strategy improvement algorithm 
we adapt in this article 
heavily relies on the fact that,
for all sequential statements $s$,
$\sem{s}^\sharp$ 
is a point-wise minimum of finitely many monotone and weak-affine operators.
The latter statement especially implies that $\sem{s}^\sharp$
is concave (see \citet{arXiv_report} for precise definitions).

The number of vertices in the feasible space of the point-wise infimum in \eqref{eqn:dual_lp} may be exponential in the size of the original problem, and thus the representation as a point-wise minimum of finitely many monotone and weak-affine operators might contain an exponential number of such operators. 
This is not a problem since our algorithm never computes this decomposition explicitly.

Any polynomial-time method for evaluating the abstract semantics of sequential statements 
can be used to derive a polynomial-time method for evaluating merge-simple statements.

\begin{lem}[Merge-Simple Statements]
\label{l:merge-simple:poly}
  Let $s$ be a merge-simple statement.
  The operator $\sem{s}^\sharp$ is a point-wise maximum 
  of finitely many point-wise minima of finitely many monotone and weak-affine mappings.
  For all $d \in \CR^m$, 
  $\sem{s}^\sharp (d)$ can be computed in polynomial time
  through linear programming.
\end{lem}

\begin{proof}
  Let $s \equiv s_1 \vee \cdots \vee s_k$,
  where  
  $s_1,\ldots,s_k$ are sequential statements.
  Since  
  $\sem{s}^\sharp (d) = \sem{s_1}^\sharp (d) \vee\cdots\vee \sem{s_k}^\sharp (d)$,
  Lemma~\ref{l:sequential:poly},
  can be applied to provide us with the desired result.
%
%
\end{proof}

\noindent
The problem for arbitrary statement is more difficult. By clear equivalence with satisfiability solving modulo the theory 
of linear real arithmetic, we obtain:

\begin{lem}
\label{l:diamand:is:np:complete}
  The problem of deciding,
  whether or not,
  for a given template constraint matrix $T$,
  and a given statement $s$,
  $
    \sem{s}^\sharp (\inftyvar) > \neginftyvar
  $
  holds,
  is \complexclass{NP}-complete.  
  \qed
\end{lem}

\subsection{A Trivial Method for Computing Abstract Semantics}

Using the results we have obtained so far,
the abstract semantics of a program $G$ w.r.t.\ some
template constraint matrix $T$ can be computed 
using the following two-step procedure:
\begin{enumerate}[(1)]
  \item
    Replace each statement $s$ of the program $G$ with an equivalent merge-simple statement.
    This corresponds to an explicit enumeration of all paths between cut-points,
    which potentially causes an exponential blowup.
  \item
    Apply the methods of \citet{DBLP:conf/csl/GawlitzaS07}
    to the obtained program to compute the abstract semantics
    $\Values^\sharp$ of $G$.
\end{enumerate}

\noindent
Because of the possible exponential blowup, 
the above described method is impractical for most cases%
\footnote{
Note that we cannot expect a polynomial-time algorithm,
because of Lemma~\ref{l:diamand:is:np:complete}: even without
loops, abstract reachability is \complexclass{NP}-hard.
Even if all statements are merge-simple,
we cannot expect a polynomial-time algorithm,
since the problem of computing the winning regions 
of parity games 
is polynomial-time reducible to abstract reachability
\cite{DBLP:conf/fm/GawlitzaS08}.}.
Our method eschews this blowup as follows:
instead of enumerating all program paths, we shall visit them only as needed.
Guided by a SAT modulo linear real arithmetic solver, 
our method selects a path through a statement $s$ 
only when it is \emph{locally profitable} in some sense.
In the worst case, an exponential number of paths may be visited
(Section~\ref{s:upper}); but one can hope that this rarely happens in practice.
In cases in which our algorithm needs exponential time, 
it at least avoids the explicit exponential expansions. 
It uses only polynomial space.


%% file: eqs.tex
\section{Max-Strategy Iteration}
\label{s:max:strat:imp}

This section presents our main contribution.
We adapt 
the max-strategy improvement schema 
of \mycitet{DBLP:journals/toplas/GawlitzaS11} to obtain an algorithm to compute abstract semantics in the framework of this article.

\subsection{Notations}

Before we go in medias res,
we have to introduce some notations.
A system $\E$ of (fixpoint) equations over $\CR$ is a finite set 
$\{ \vx_1 = e_1, \ldots, \vx_n = e_n \}$
of equations.
Here, 
$\vx_1,\ldots,\vx_n$ are pairwise distinct, $\CR$-valued variables and 
$e_1,\ldots,e_n$ are expressions over $\CR$.
We denote the set $\{\vx_1,\ldots,\vx_n\}$ of variables of $\E$ by $\vX_\E$.
We omit the subscript, 
whenever it is clear from the context.
A function $\rho : \vX\to\CR$ is called a \emph{variable assignment}. 
It assigns the value $\rho(\vx)$ to each variable $\vx \in \vX$.
Variable assignments are ordered by the point-wise extension of $\leq$ on $\CR$,
i.e., $\rho \leq \rho'$ if and only if $\rho(\vx) \leq \rho'(\vx)$ for all $\vx \in \vX$.
Since $\CR$ is a complete linearly ordered set,
the set $\vX\to\CR$ of all variable assignments is a complete lattice.
The semantics $\sem{e} : (\vX\to\CR)\to\CR$ of an expression $e$
is defined by 
$\sem{\vx}(\rho) := \rho(\vx)$ and 
$\sem{f(e_1,\ldots,e_k)}(\rho) := f(\sem{e_1}(\rho),\ldots,\sem{e_k}(\rho))$,
where $\vx \in \vX$, $f$ is a $k$-ary operator on $\CR$,
$e_1,\ldots,e_k$ are
expressions,
and $\rho : \vX\to\CR$ is a variable assignment.
We define the operator 
$\sem\E : (\vX\to\CR) \to \vX\to\CR$
by
$\sem\E(\rho)(\vx) := \sem{e}\rho$ for all equations 
$\vx = e$ of $\E$, all $\rho : \vX\to\CR$, and all $\vx \in \vX$.
%
A fixpoint equation $\vx = e$ is called \emph{monotone} if and only if  
all operators used in $e$ are monotone.
Then, the evaluation function $\sem{e}$ of $e$ is monotone, too. 
Finally, 
the operator $\sem\E$ is monotone
for all systems $\E$ of monotone (fixpoint) equations.
A variable assignment $\rho$ is called a 
\emph{solution} (resp.\ \emph{pre-solution}, resp.\ \emph{post-solution})
of $\E$ 
if and only if $\rho = \sem\E(\rho)$
(resp.\ $\rho \leq \sem\E(\rho)$, resp.\ $\rho \geq \sem\E(\rho)$).
The least solution of $\E$ is denoted by
$\mu\sem\E$.
If the operator $\sem\E$ is monotone,
then the fixpoint theorem of Knaster/Tarski  (Theorem~\ref{th:tarski}) ensures the 
existence of a uniquely determined least solution $\mu\sem\E$. 
For a system $\E$ of equations and a pre-solution $\rho$,
$\mu_{\geq \rho}\sem\E$ denotes the least solution of $\E$ 
among the solutions of $\E$ that are greater than or equal to $\rho$, i.e.,
$\mu_{\geq \rho}\sem\E = \min \{ \rho' \mid \rho' = \sem\E(\rho') \text{ and } \rho' \geq \rho \}$.
Again,
if the operator $\sem\E$ is monotone,
then the fixpoint theorem of Knaster/Tarski ensures the 
existence of $\mu_{\geq\rho}\sem\E$,
since the set $\{ \rho' \mid \rho' \geq \rho \}$ is a complete lattice.

\subsection{Rewriting the Abstract Semantic Equations}
\label{ss:eqs}

The first step of our method consists of rewriting our static analysis problem 
into a system of monotone fixpoint equations over $\CR$.
Assume that 
$G = (N,E,\start)$ is a program that has $n$ variables, 
and $T \in \R^{m\times n}$ is a template constraint matrix.
%
Recall that (w.r.t.\ $T$) the abstract semantics of $G$ is the least solution
of the following constraint system
(cf.\ \eqref{eq:sem:constraints} in Subsection~\ref{ss:abs:sem}):
\begin{align}
  \VALUES^\sharp[\start]
    &\geq \alpha(\R^n)
  &
  \VALUES^\sharp[v] 
    \geq \sem{s}^\sharp (\VALUES^\sharp[u])
  && \text{for all } (u,s,v) \in E
\end{align}

\noindent
The constraint system has exactly one $\CR^m$-valued variable $\VALUES^\sharp[v]$ for each program point $v \in N$.
For each program point $v \in N$,
we decompose the $\CR^m$-valued variable $\VALUES^\sharp[v]$ into $m$ $\CR$-valued variables
$\vd_{v,1},\ldots,\vd_{v,m}$.
That is, 
we set 
$(\vd_{v,1},\ldots,\vd_{v,m})^\top = \VALUES^\sharp[v]$.
We obtain the following constraint system:
\begin{align}
  \vd_{\start,i} &\geq \infty 
    && \text{for all } i \in \{1,\ldots,m\} \\
  \vd_{v,i}        &\geq \sem{s}^\sharp_{i\cdot} \left( \vd_{u,1},\ldots,\vd_{u,m} \right)
    && 
      \text{for all } (u,s,v) \in E
      \text{ and all } i \in \{1,\ldots,m\} 
\end{align}

\noindent
The fixpoint theorem of Knaster/Tarski (Theorem~\ref{th:tarski})
ensures that the least solution of the above system of inequalities 
is the least solution of the following equation system:
\begin{align}
  \label{eq:abs:sem:decomp:1}
  \vd_{\start,i} &\!=\! \infty 
    && 
       \!\!\!
       \text{for all } i \in \{1,\ldots,m\} \\
  \label{eq:abs:sem:decomp:2}
  \vd_{v,i}        &\!=\! 
    \max \left\{ 
      \sem{s}^\sharp_{i\cdot} \left( \vd_{u,1},\ldots,\vd_{u,m} \right)
      \mid
      (u,s,v) \in E
    \right\}
    && 
      \!\!\!
      \text{for all } v \in N \setminus \{ \start \},
      i \in \{1,\ldots,m\}   
\end{align}

\noindent
We denote the above system of fixpoint equations by $\E(G,T)$.
From Section~\ref{s:basic_obs},
we know that the right-hand sides of $\E(G,T)$ 
are point-wise maxima of finitely many point-wise minima of finitely many weak-affine operators.
We summarize the properties of $\E(G,T)$:

\begin{lem}
  \label{l:eqs:abs_sem}
  Let $G$ be a program and $\Values^\sharp$ its abstract semantics
  (w.r.t.\ the template constraint matrix $T \in \R^{m\times n}$).
  Let $\rho^* := \mu\sem{\E(G,T)}$ be the least solution of $\E(G,T)$.
  Then 
  $\Values^\sharp_{i \cdot}[v] = \rho^*(\vd_{v,i})$
  for all 
  program points $v \in N$ and all $i \in \{1,\ldots,m\}$.
  The right-hand sides of $\E(G,T)$ 
  are point-wise maxima of finitely many point-wise minima of finitely many weak-affine operators.
  Thus, they are in particular point-wise maxima of finitely many monotone and concave functions.
  \qed
\end{lem}

\begin{figure}
  \centering
  \begin{tabular}{@{}cc@{}}
    $
    \begin{array}{@{}r@{\;}l@{}}
    G &= (N,E,\start) 
    \\[3pt]
    N &= \{ \start, 1 \}
    \\[3pt]
    E &= \{ (\start, \vx_1' = 0, 1), (1, s, 1) \}
    \\[3pt]
    s      &\equiv \Phi \wedge (\Phi_1 \vee \Phi_2 )
    \\[3pt]    
    \Phi     &\equiv  \vx_1 \leq 1000 \wedge \vx_2' = -\vx_1  
    \\[3pt]
    \Phi_1 &\equiv \vx_2' \leq -1 \wedge \vx_1' = -2 \vx_1 
    \\[3pt]
    \Phi_2 &\equiv -\vx_2' \leq 0 \wedge \vx_1' = - \vx_1 + 1 
    \end{array}
    $
    &
      \begin{tabular}[c]{c}
        \qquad
	\scalebox{1}{
	\begin{tikzpicture}
		 \node (start) [point] {$\start$};
		 \node (n1) [below of = start,point,yshift=1mm]{$1$};
		 \node (n2) [coordinate,below of = n1,yshift=0mm]{$2$};
		 \node (n3) [coordinate,left of = n2]{$2$};
		 \node (n4) [coordinate,left of = n1]{$2$};
		 \path[->] (start) edge [] node [right] {$\vx_1' = 0$} (n1);
		 \path[-] (n1) edge [] node [right] 
		   {$s$} 
		   (n2);
		 \path[-] (n2) edge [] node [right] {} (n3);
		 \path[-] (n3) edge [] node [right] {} (n4);
		 \path[->] (n4) edge [] node [right] {} (n1);
	\end{tikzpicture}
	}
      \end{tabular}
  \end{tabular}
  
  \smallskip
  (a) The program $G$ 
  
  \begin{align*}
    T = \begin{pmatrix} 1 & 0 \\ -1 & 0 \end{pmatrix}
  \end{align*}
  
  \smallskip
  (b) The template constraint matrix $T$\\
(only $x_1$ is taken into account in the template, thus the zero right column)

  \begin{align*}
    \vd_{\start,1} &= \infty &
    \vd_{1,1} &= 
    \max\,\left\{ 
      \sem{\vx_1' = 0}^\sharp_{1\cdot} (\vd_{\start,1},\vd_{\start,2}) ,\; 
      \sem{s}^\sharp_{1\cdot} (\vd_{1,1},\vd_{1,2}) 
    \right\}
    \\
    \vd_{\start,2} &= \infty &
    \vd_{1,2} &= \max\,\left\{ 
      \sem{\vx_1' = 0}^\sharp_{2\cdot} (\vd_{\start,1},\vd_{\start,2}) ,\; 
      \sem{s}^\sharp_{2\cdot} (\vd_{1,1},\vd_{1,2})
    \right\} 
  \end{align*}
  
  \smallskip
  (c) The equation system $\E(G,T)$

  \caption{The running example}
  \label{fig:run:ex}
  \bigskip
\end{figure}

\begin{exas}
  \label{ex:running:1}
  We again consider our running example 
  specified in Figure~\ref{fig:run:ex}(a).
  We want to perform the analysis w.r.t.\ the template constraint matrix $T$ specified 
  in Figure~\ref{fig:run:ex}(b).
  The resulting equation system $\E(G,T)$ 
  is shown in Figure~\ref{fig:run:ex}(c).

    The least solution $\rho^* := \mu\sem{\E(G,T)}$ of $\E(G,T)$ is given by
    $ 
      \rho^*
      = 
      \{ \vd_{\start,1} \mapsto \infty, \vd_{\start,2} \mapsto \infty, \vd_{1,1} \mapsto 2001, \vd_{1,2} \mapsto 2000 \}
    $. 
	Thus, by Lemma~\ref{l:eqs:abs_sem}, 
	$\Values^\sharp[\start] = (\infty,\infty)$, 
	and
	$\Values^\sharp[1] = (2001,2000)$.
	In consequence,
	all possible values of the program variable $\vx_1$ at program point $1$ 
	are in the interval $[-2000,2001]$.
  \qed
\end{exas}


%% file: max_strat_imp.tex

\subsection{Adapting the Max-Strategy Improvement Algorithm}

Following the lines of \citet{arXiv_report},
our starting point is a system $\E$ of monotone fixpoint equations of the form
$\vx = \max \,  \Sigma_\vx$,
where $\vx$ is a $\CR$-valued variable, 
and $\Sigma_{\vx}$ is a finite set of monotone and concave expressions over $\CR$.
An expression $e$ is called monotone (resp.\ concave) if and only if
$\sem{e}$ is monotone (resp.\ concave).%
\footnote{%
  For a precise definition of concavity for functions from the set $\CR^n\to\CR^m$,
  we refer to \citet{techrep_conc}.
  For this article, 
  however,
  a precise treatment of these issues is not required.
  We just mention concavity to give a better intuition.
}
We treat a function from the finite set $\vX$ 
of variables to $\CR$ as a vector of $|\vX|$ elements from $\CR$.
%
In our application
---
recall that we aim at solving the equation system $\E(G,T)$
---
the sets $\Sigma_\vx$ are implicitly and succinctly given 
by the right-hand sides of equations of the forms 
\eqref{eq:abs:sem:decomp:1} and \eqref{eq:abs:sem:decomp:2}.
Indeed, every expression of the form
$\sem{s}^\sharp_{i\cdot}(\vd_{u,1},\ldots,\vd_{u,m})$,
found on the right-hand side of such equations,
can be equivalently rewritten into  
$ 
  \max\,\{
    \sem{s_1}^\sharp_{i\cdot}(\vd_{u,1},\ldots,\vd_{u,m}),
    \allowbreak
    \ldots, \allowbreak
    \sem{s_k}^\sharp_{i\cdot}(\vd_{u,1},\ldots,\vd_{u,m})    
  \},
$ 
%
where $s_1,\ldots,s_k$ are (potentially exponentially many) sequential statements.
Since $s_1,\ldots,s_k$ are sequential,
the operators  
$\sem{s_1}^\sharp_{i\cdot}, \allowbreak \ldots, \allowbreak \sem{s_k}^\sharp_{i\cdot}$ 
are point-wise minima of finitely many monotone and weak-affine operators; 
hence they are monotone and concave operators.

One obvious way to solve the system $\E$ of equations is
to perform the above mentioned rewriting explicitly and then apply the 
max-strategy improvement algorithm.
To avoid this impractical exponential blowup,
in what follows we modify the algorithm such that it directly works on the succinct representation.



Assume that $\E$ denotes a system of fixpoint equations of the form $\vx = \max\,\Sigma_\vx$,
where $\Sigma_\vx$ is a finite set of monotone and concave expressions over $\CR$.
A \emph{max-strategy} $\sigma$ for $\E$ 
is a system of equations such that, for each equation $\vx = e$ of $\sigma$, 
one of the following statements holds:
\begin{enumerate}[(1)]
  \item
    $e$ is $\neginfty$.
  \item
    $e \in \Sigma_\vx$,
    where 
    $\vx = \max \, \Sigma_\vx$
    is an equation of $\E$.
\end{enumerate}

\noindent
Intuitively, a max-strategy picks for each maximum operator one of its operands.
For a system $\E$ of equations,
we denote the set of all max-strategies by $\Sigma_\E$.
In our application,
the cardinality of $\Sigma_\E$ is exponential in the size of $\E$.
To be more precise,
it is in $\mathcal O (2^{n^2})$, 
where $n$ denotes the size of $\E$.
Enumerating all max-strategies is therefore impractical.

\begin{exas}
  \label{ex:running:2}
  We continue our running example (Figure~\ref{fig:run:ex}).
  Consider the system  $\E(G,T)$ and note that 
  $s \equiv \Phi \wedge (\Phi_1 \vee \Phi_2) \equiv (\Phi \wedge \Phi_1) \vee (\Phi \wedge \Phi_2)$;
  therefore the equation
  \begin{align}
   \vd_{1,1} 
   = 
    \max\,\left\{ 
      \sem{\vx_1' = 0}^\sharp_{1\cdot} (\vd_{\start,1},\vd_{\start,2}) ,\; 
      \sem{s}^\sharp_{1\cdot} (\vd_{1,1},\vd_{1,2}) 
    \right\}
  \end{align}
  
  \noindent
  can be equivalently rewritten into
  \begin{align}
   \vd_{1,1} 
   = 
    \max\,\left\{ 
      \sem{\vx_1' = 0}^\sharp_{1\cdot} (\vd_{\start,1},\vd_{\start,2}) ,\; 
      \sem{\Phi \wedge \Phi_1}^\sharp_{1\cdot} (\vd_{1,1},\vd_{1,2}) ,\;
      \sem{\Phi \wedge \Phi_2}^\sharp_{1\cdot} (\vd_{1,1},\vd_{1,2}) 
    \right\}
    .
  \end{align}
  \noindent Recall that this expansion is solely for the purpose of proving properties: it is not done in the algorithm.
  The equation system $\sigma$ consisting of the equations
  \begin{align}
    \vd_{\start,1} &= \infty &
    \vd_{1,1} &= \sem{\Phi \wedge \Phi_2}^\sharp_{1\cdot} (\vd_{1,1},\vd_{1,2}) &
    \vd_{\start,2} &= \infty &
    \vd_{1,2} 
      &= \sem{\vx_1' = 0}^\sharp_{2\cdot} (\vd_{\start,1},\vd_{\start,2}) 
  \end{align}
  \noindent 
  is thus a max-strategy for this system.
  \qed
\end{exas}


\noindent
A crucial notion we need in the following is the notion of improvements.
Let $\sigma$ be a max-strategy for $\E$ and $\rho$ a pre-solution of $\sigma$.
A max-strategy $\sigma'$ for $\E$ is called an 
\emph{improvement of $\sigma$ w.r.t.\ $\rho$}
if and only if 
the following conditions are fulfilled:
\begin{enumerate}[(1)]
      \item
			If $\rho \neq \sem\E(\rho)$,
			then $\sem{\sigma'}(\rho) > \rho$.
      \item
			If $\vx = e$ is an equation of $\sigma$
			and
			$\vx = e'$ is an equation of $\sigma'$
			with $e \neq e'$,
			then
			$\sem{e'}(\rho) > \sem{e}(\rho)$.
\end{enumerate}

\begin{exa}
  We continue our running example (Figure~\ref{fig:run:ex}).
  We consider the equation system $\sigma'$ that consists of the following equations:
  \begin{align}
    \vd_{\start,1} &= \infty &
    \vd_{1,1} &= \sem{\Phi \wedge \Phi_2}^\sharp_{1\cdot} (\vd_{1,1},\vd_{1,2}) &
    \vd_{\start,2} &= \infty &
    \vd_{1,2} 
      &= \sem{\Phi \wedge \Phi_1}^\sharp_{2\cdot} (\vd_{1,1},\vd_{1,2})     
  \end{align}
  
  \noindent
  The equation system $\sigma'$ is a max-strategy of $\E(G,T)$ and moreover
  an improvement of the max-strategy $\sigma$ (defined in Example~\ref{ex:running:2})
  w.r.t.\ the variable assignment
   \begin{align}
     \rho
    := 
    \{ \vd_{\start,1} \mapsto \infty, \vd_{\start,2} \mapsto \infty, \vd_{1,1} \mapsto 1, \vd_{1,2} \mapsto 0 \}
    .
  \end{align}
  
  \noindent
%
  It is an improvement,
  since  
  $
    \sem{\sem{\Phi \wedge \Phi_1}^\sharp_{2\cdot} (\vd_{1,1},\vd_{1,2})} (\rho)
    =
    1
    >
    0
    =
    \sem{\sem{\vx_1' = 0}^\sharp_{2\cdot}(\vd_{1,1}, \vd_{1,2})} (\rho)
  $. 
  In this example, 
  $\sigma'$ is the only improvement of $\sigma$ w.r.t.\ $\rho$.
  \qed
\end{exa}

\noindent
Note that, for a max-strategy $\sigma$ and a pre-solution $\rho$ of $\sigma$, 
there might be several max-strategies $\sigma'$ that are improvements of $\sigma$ w.r.t.\ $\rho$.
Consider, for instance, the equation system 
$\E = \{ \vx = \max \, \{ 0 ,\, 1 ,\, 2 \} \}$.
Both, 
the max-strategies $\{ \vx = 1 \}$ and $\{ \vx = 2 \}$ are improvements 
of the max-strategy $\{ \vx = 0 \}$.
For the results we are going to develop in this article,
it is not important which improvement we choose: this will neither affect
the final result obtained, nor change the worst-case complexity bounds that
we prove. It is however possible that different heuristics may
lead to different practical complexities.


The max-strategy improvement algorithm starts with
the max-strategy 
$\sigma_0 := \{ \vx = \neginfty \mid \vx \in \vX \}$ 
and the variable assignment 
$\rho_0 := \{ \vx \mapsto \neginfty \mid \vx \in \vX \}$.
The algorithm successively performs the following two 
steps in the given order until it has found the least solution:
  \begin{enumerate}[(1)]
    \item
      Improve the max-strategy $\sigma$ w.r.t.\ $\rho$.
    \item
      Evaluate the max-strategy $\sigma$ w.r.t.\ $\rho$ to obtain a new value for $\rho$. 
  \end{enumerate}

\noindent
In pseudo-code, we can formulate it as follows:

\begin{algorithm}[H]
    $
    \begin{array}{rl}
    1: & \sigma \GETS \sigma_{0} ; \\
    2: &\rho \GETS \rho_{0} ; \\
    \VS
	3: &\WHILE (\rho < \sem\E(\rho) ) \;
	\{ \\
		4: &\hspace*{0.5cm} \sigma \GETS \text{improvement of $\sigma$ w.r.t.\ $\rho$}; \\ 
		5: &\hspace*{0.5cm}\rho \GETS \mu_{\geq \rho} \sem{\sigma} ; \\
	6: &\} \\
	\VS
	7: &\RETURN \rho;
	\\
     \end{array}
     $\hfill
	\caption{The Max-Strategy Improvement Algorithm}
	\label{alg:alg:stratimp}
\end{algorithm}

  \noindent
    For all $i \in \N$, 
    let
    $\rho_i$ be the value of the  variable $\rho$ 
    and
    $\sigma_i$ be the value of the variable $\sigma$ 
    after the $i$-th evaluation of the loop-body.
  %
We have:
\begin{lem}[\cite{techrep_conc}, {\cite[Lem.~6.7]{DBLP:journals/toplas/GawlitzaS11}}]
  \label{l:alg:sequence}
	The following statements hold for all $i \in \N$:
	\begin{enumerate}[\em(1)]
	  \item
			$\rho_i \leq \mu\sem\E$.
      \item
			$\rho_i \leq \sem{\sigma_{i+1}} (\rho_i)$.
      \item
			If $\rho_i < \mu\sem\E$, then $\rho_{i+1} > \rho_i$.
      \item
			If $\rho_i = \mu\sem\E$, then $\rho_{i+1} = \rho_i$.
			\qed
	\end{enumerate}
\end{lem}

\noindent
The above lemma implies that 
the algorithm returns the least solution,
whenever it terminates.
Whether or not it terminates
depends on the properties of the class of fixpoint equation systems under consideration.
In our application,
we aim at computing the least solution of the equation system 
$\E(G,T)$ (see Subsection~\ref{ss:eqs}).
By Lemma~\ref{l:eqs:abs_sem},
the right-hand sides of $\E(G,T)$ are point-wise maxima of finitely many monotone and concave functions.
More specifically, 
the right-hand sides are point-wise maxima of finitely many point-wise minima of finitely many weak-affine operators.
This property guaranties the termination of the 
max-strategy improvement algorithm \cite{techrep_conc}\cite[§6.1]{DBLP:journals/toplas/GawlitzaS11}.
At the latest, 
it terminates after considering each max-strategy at most linearly often (see Lemma \ref{l:term:after:lin}).
Before we explain the remaining building blocks,
i.e., how to execute program lines 4 and 5,
we consider an example.

\begin{exa}
  \label{ex:run:ex:lag:run}
  We consider our running example.
  That is,
  we aim at computing the least solution of the equation system 
  $\E(G,T)$ shown in Figure~\ref{fig:run:ex}.
  Running the algorithm can, for instance, give us the following trace:
  \begin{align}
    \sigma_0 
    &:= 
    \{ 
      \vd_{\start,1} = \neginfty, \,
      \vd_{\start,2} = \neginfty, \,
      \vd_{1,1} = \neginfty, \,
      \vd_{1,2} = \neginfty 
    \}
    \\
    \rho_0 
    &:= 
    \{ 
      \vd_{\start,1} \mapsto \neginfty, \,
      \vd_{\start,2} \mapsto \neginfty, \,
      \vd_{1,1} \mapsto \neginfty, \,
      \vd_{1,2} \mapsto \neginfty 
    \}
    \\[4pt]
    \sigma_1 
    &:= 
    \{ 
      \vd_{\start,1} = \infty, \,
      \vd_{\start,2} = \infty, \,
      \vd_{1,1} = \neginfty, \,
      \vd_{1,2} = \neginfty
    \}
    \\
    \rho_1 
    &:= 
    \{ 
      \vd_{\start,1} \mapsto \infty, \,
      \vd_{\start,2} \mapsto \infty, \,
      \vd_{1,1} \mapsto \neginfty, \,
      \vd_{1,2} \mapsto \neginfty
    \}
    \\[4pt]
    \sigma_2
    &:= 
    \{ 
      \vd_{\start,1} = \infty, \,
      \vd_{\start,2} = \infty, \,
      \vd_{1,1} = \sem{\vx_1' = 0}^\sharp_{1\cdot} (\vd_{\start,1}, \vd_{\start,2}), \,
      \\&\qquad
      \vd_{1,2} = \sem{\vx_1' = 0}^\sharp_{2\cdot} (\vd_{\start,1}, \vd_{\start,2})
    \}
    \\
    \rho_2
    &:= 
    \{ 
      \vd_{\start,1} \mapsto \infty, \,
      \vd_{\start,2} \mapsto \infty, \,
      \vd_{1,1} \mapsto 0, \,
      \vd_{1,2} \mapsto 0
    \}
    \\[4pt]
    \sigma_3
    &:= 
    \{ 
      \vd_{\start,1} = \infty, \,
      \vd_{\start,2} = \infty, \,
      \vd_{1,1} = \sem{\Phi \wedge \Phi_2}^\sharp_{1\cdot} (\vd_{1,1},\vd_{1,2}), \,
      \\&\qquad
      \vd_{1,2} = \sem{\vx_1' = 0}^\sharp_{2\cdot} (\vd_{\start,1},\vd_{\start,2}) 
    \}
    \\
    \rho_3
    &:= 
    \{ 
      \vd_{\start,1} \mapsto \infty, \,
      \vd_{\start,2} \mapsto \infty, \,
      \vd_{1,1} \mapsto 1, \,
      \vd_{1,2} \mapsto 0
    \}
    \\[4pt]
    \sigma_4
    &:= 
    \{ 
      \vd_{\start,1} = \infty, \,
      \vd_{\start,2} = \infty \,
      \vd_{1,1} = \sem{\Phi \wedge \Phi_2}^\sharp_{1\cdot} (\vd_{1,1},\vd_{1,2}), \,
      \\&\qquad
      \vd_{1,2} = \sem{\Phi \wedge \Phi_1}^\sharp_{2\cdot} (\vd_{1,1},\vd_{1,2}) 
    \}
    \\
    \rho_4
    &:= 
    \{ 
      \vd_{\start,1} \mapsto \infty, \,
      \vd_{\start,2} \mapsto \infty, \,
      \vd_{1,1} \mapsto 2001, \,
      \vd_{1,2} \mapsto 2000
    \}
  \end{align}
  
  \noindent
  Here, for all $i$, $\rho_{i+1} = \mu_{\geq \rho_{i}} \sem{\sigma_{i+1}}$
  and $\sigma_{i+1}$ is an improvement of $\sigma_i$ w.r.t.\ $\rho_i$.
%
The variable $\rho_4$ is a solution of $\E(G,T)$.
The max-strategy improvement algorithm terminates 
with the correct least solution, which is $\rho_4$.
  \qed
\end{exa}

\noindent
We now present methods to evaluate max-strategies 
(Line 5 of Algorithm~\ref{alg:alg:stratimp})
and 
to improve max-strategies 
(Line 4 of Algorithm~\ref{alg:alg:stratimp}).

\subsection{Evaluating Max-Strategies}
\label{ss:veal:max:strat}

We restrict our consideration to our application.
That is, 
we assume that 
the equation system $\E$ is given by
$\E = \E(G,T)$ for some program $G$ and some template constraint matrix $T$. 
For all $i \in \N$,
this allows us to compute $\rho_i$ 
as follows:


\begin{lem}[\cite{techrep_conc},\cite{DBLP:journals/toplas/GawlitzaS11}]
  \label{l:how:to:eval:strats}
  Let $i \in \N$.
  Recall that, by construction, $\rho_{i+1} = \mu_{\geq\rho_i}\sem{\sigma_{i+1}}$.
  The variable assignment $\rho_{i+1}$ can be computed as follows:
  Let $\E'$ denote the system of equations that is
  obtained from the equation system $\sigma_{i+1}$ by performing the following steps:
  \begin{enumerate}[\em(1)]
    \item 
      Remove every equation $\vx = e$, 
      where $\sem e (\rho_i) = \neginfty$
      and replace then the remaining occurrences of $\vx$ with the constant $\neginfty$.
    \item 
      Remove every equation $\vx = e$, 
      where $\sem e (\rho_i) = \infty$
      and replace then the remaining occurrences of $\vx$ with the constant $\infty$.
  \end{enumerate}
  
  \noindent
  For all equations $\vx = e$ of the equation system $\sigma_{i+1}$ 
  with $\neginfty < \sem{e} (\rho_i) < \infty$,
  we can compute $\rho_{i+1}(\vx)$ as follows:
  \begin{align}
    \label{eq:l:eindeutige:loesung:5}
    \rho_{i+1}(\vx)
    =
    \sup\; \{ \rho(\vx) \mid \rho:\vX_{\E'} \to\R ,\; \rho \leq \sem{\E'}(\rho)\}
  \end{align}

  \noindent
  The value $\rho_{i+1}$ only depends on the equation system $\sigma_{i+1}$ and the 
  set of variables already identified to be $\infty$,
  namely, 
  $\{ \vx \mid \vx = e \text{ is an equation of } \sigma_{i+1} \text{ with } \sem e (\rho_{i}) = \infty \}$.
  \qed  
\end{lem}

\noindent
  In consequence,
  the max-strategy improvement algorithm has to consider 
  each max-strategy at most $\abs\vX$ times.
  Hence, we have:

\begin{lem}[\cite{techrep_conc},\cite{DBLP:journals/toplas/GawlitzaS11}]
  \label{l:term:after:lin}
  The max-strategy improvement algorithm
  terminates after at most $\abs\vX \cdot \abs{\Sigma_\E}$ max-strategy improvement steps.
  \qed
\end{lem}

%
%
%
%

\noindent
Lemma~\ref{l:how:to:eval:strats} gives us a method for computing $\rho_i$.
For each variable $\vx \in \vX$, 
we have to compute 
\begin{align}
  \label{eq:sup:to:comp}
  \sup \, \left\{
    \rho(\vx) \mid \rho : \vX_{\E'} \to\R \text{ and } \rho \leq \sem{\E'} (\rho)
  \right\}
  .
\end{align}

\noindent
The equations of $\E'$
are of the form 
$\vb = \sem{s}^\sharp_{k\cdot}(\vb_1,\ldots,\vb_m)$,
where
$\vb, \vb_1,\ldots,\vb_m$ are $\CR$-valued variables, 
and $s$ is a sequential statement.
Thus, by Lemma~\ref{l:sequential:poly},
the right-hand sides are point-wise minima of finitely many monotone and weak-affine functions.
Hence, they are monotone and concave.
Therefore, 
\eqref{eq:sup:to:comp}
represents a convex optimization problem.

The above convex optimization problem is of a very special form.
The right-hand sides are parameterized linear programs.
In consequence,
the convex optimization problem can be rewritten 
into an equivalent linear programming problem as follows:
In accordance to~\eqref{eq:this:is:an:lp:1} and \eqref{eq:this:is:an:lp:2},
in $\E'$,
we replace each equation $\vb = \sem{s}^\sharp_{k\cdot}(\vb_1,\ldots,\vb_m)$ 
with the following linear constraints:
\begin{align}
  \vb &\leq T_{k\cdot} (\vy_1', \ldots, \vy_n')^\top \\
  & \;\;\Phi \\
  T (\vy_1,\ldots,\vy_n)^\top &\leq (\vb_1,\ldots,\vb_m)
\end{align}

\noindent
Here, $\vy_1,\ldots,\vy_n,\vy_1',\ldots,\vy_n'$ are fresh variables.
$\Phi$ is a set of linear inequalities that 
is obtained from the sequential statement $s$ by 
\begin{enumerate}[(1)]
  \item
    replacing the variables $\vx_1,\ldots,\vx_n,\vx_1',\ldots,\vx_n'$ 
    with the fresh variables $\vy_1,\ldots,\vy_n,\vy_1',\allowbreak\ldots,\vy_n'$,
  \item
    replacing all other variables of $s$ with fresh variables, and
  \item
    replacing every strict inequality $<$ with a non-strict inequality $\leq$.
\end{enumerate}
We denote the resulting constraint system by $\mathcal C$.
By construction, we have:
\begin{align}
  \sup \, \left\{
    \rho(\vx) \mid \rho : \vX\to\R \text{ and } \rho \leq \sem{\E'} (\rho)
  \right\}
  =
  \sup \, \left\{
    \rho(\vx) \mid \rho : \vX\to\R \text{ and } \rho \text{ solves } \mathcal C
  \right\}
\end{align}

\noindent
The construction can be carried out in polynomial time.
Since $\mathcal C$ is a set of linear constraints,
we can use linear programming to compute the optimal value.
We have:

\begin{lem}[Evaluating Max-Strategies]
  Whenever our max-strategy improvement algorithm has to compute 
  $\mu_{\geq\rho}\sem{\sigma}$,
  this can be performed by solving 
  $\abs\vX$ linear programming problems of polynomial size.
  The linear programming problems 
  do only depend on
  $\sigma$ and the set 
  $\{ \vx \mid \vx = e \text{ is an equation of } \sigma \text{ with } \sem e (\rho) = \infty \}$.
  \qed  
\end{lem}

\begin{exa}
	We now discuss how to compute
	$
	  \rho_3
	  := 
	  \mu_{\geq\rho_2}\sem{\sigma_3}
	$
	from Example~\ref{ex:run:ex:lag:run}.
         %
	Note that the values of the variables $\vd_{\start,1}$ and $\vd_{\start,2}$ 
	are already known to be $\infty$.
	It remains to determine the values for the variables $\vd_{1,1}$ and $\vd_{1,2}$.
	According to Lemma~\ref{l:how:to:eval:strats},
	we have 
	\begin{align}
	  \nonumber
	  \rho_3(\vd_{1,1})
	  &=
           \sup\; \{ 
             \vd_{1,1} 
             \mid 
               \vd_{1,1}, \vd_{1,2} \in \R ,\,
               \vd_{1,1} \leq \sem{\Phi \wedge \Phi_2}^\sharp_{1\cdot} (\vd_{1,1}, \vd_{1,2}), \,
               \\&\qquad
               \vd_{1,2} \leq \sem{\vx_1' = 0}^\sharp_{2\cdot} (\infty, \infty) 
             \}
	\end{align}
	
         \noindent
	Observe that 
	$\Phi \wedge \Phi_2$ 
	can be equivalently rewritten into
	$\vx_1 \leq 0 \wedge \vx_1' = -\vx_1 + 1 \wedge \vx_2' = - \vx_1$.
	Thus, according to the above observations,
	$\rho_3(\vd_{1,1})$
	is the optimal value of the following linear programming problem:
	  \begin{align}
	    \max \; & \vd_{1,1}
	    &
	    \vd_{1,1} &\leq -\vx_1 + 1 
	    &
	    \vx_1       &\leq 0 
	    &
	    \vx_1       &\leq \vd_{1,1} 
	    &
	    -\vx_1      &\leq \vd_{1,2} 
	    &
	    \vd_{1,2} &\leq 0
	  \end{align}
	
	\noindent
	Since the optimal value is $1$,
	we get $\rho_3(\vd_{1,1}) = 1$.
	Similarly,
	to compute $\rho_3(\vd_{1,2})$,
	we compute the optimal value of the following linear programming problem:
	  \begin{align}
	    \max \; & \vd_{1,2}
	    &
	    \vd_{1,1} &\leq -\vx_1 + 1 
	    &
	    \vx_1       &\leq 0 
	    &
	    \vx_1       &\leq \vd_{1,1} 
	    &
	    -\vx_1      &\leq \vd_{1,2} 
	    &
	    \vd_{1,2} &\leq 0
	  \end{align}
	  
	  \noindent
	  This gives us $\rho_3(\vd_{1,2}) = 0$.
	  
	  Both linear programming problems have the
	  same feasible space. 
	  This can be utilized in an implementation to improve the performance.
	  Furthermore, 
	  $\rho_3(\vd_{1,1}) = \vd_{1,1}^*$
	  and
	  $\rho_3(\vd_{1,2}) = \vd_{1,2}^*$
	  for any optimal solution 
	  $(\vd_{1,1}^*, \vd_{1,2}^*, \vy_1^*)$
	  of the following linear programming problem:
	  \begin{align}
	    \max \; & \vd_{1,1} + \vd_{1,2}
	    &
	    \vd_{1,1} &\leq -\vx_1 + 1 
	    &
	    \vx_1       &\leq 0 
	    &
	    \vx_1       &\leq \vd_{1,1} 
	    &
	    -\vd_1      &\leq \vd_{1,2} 
	    &
	    \vd_{1,2} &\leq 0
	  \end{align}
	  
	  \noindent
	  Hence, for this example,
	  it is sufficient to solve one linear programming problem 
	  to determine the variable assignment $\rho_3$.
   \qed
\end{exa}

\noindent
The technique for evaluating max-strategies can thus be further optimized.
It is not necessary to solve one linear program for each variable.
Instead, it is possible to evaluate a max-strategy entirely
by solving only two linear programming problems 
of linear size.
The solution of the first linear programming problem
tells us which variables are to set to $\infty$.
The solution of the second linear programming problem provides us with the 
values of the variables which receive finite values.
In this article, we do not elaborate on these techniques.



%% file: improve.tex
\subsection{Improving Max-Strategies} 
\label{ss:improve}

We now discuss how we can compute an improvement of a max-strategy $\sigma$
w.r.t.\ a variable assignment $\rho$.
Since, 
by Lemma~\ref{l:diamand:is:np:complete},
this problem is \complexclass{NP}-hard,
we cannot expect to come up with a polynomial time algorithm.
We propose a solution that utilizes SMT solving techniques.

Let us first explain the intuition of our method, which is very similar to how the ``path focusing'' technique from \citet{Monniaux_Gonnord_SAS11} selects the next iteration path.
A strategy needs improvement if and only if its value does not define an inductive invariant. 
In other words: 
there is an outgoing transition from the ``invariant candidate'' into its complement, meaning that there is an execution trace through a statement, starting from the invariant candidate and ending with a violation of the current bounds.
Whether this holds is a SAT problem modulo (SMT) the theory of linear real arithmetic;
it can therefore be solved by SMT-solvers.
Furthermore, the solution from the SMT problem picks one of the sequential statements from the merge-simple expansion of the 
statement as ``offending'', explaining why the invariant candidate is not an invariant; in other words, 
it points to a possible improvement in the strategy.
More generally, the set of solutions of the SMT problem maps to the possible improvements.

Let us now see this process more formally.
Assume that we have to improve a given max-strategy 
\begin{align}
  \sigma = \{ \vx_1 = \sigma_1 ,\ldots, \vx_n = \sigma_n \}
\end{align}

\noindent
for the equation system 
\begin{align}
  \E = \{ \vx_1 = e_1 ,\ldots, \vx_n = e_n \}
\end{align}

\noindent
w.r.t.\ a variable assignment $\rho$,
which is a solution of $\sigma$,
i.e., 
$\rho = \sem{\sigma}(\rho)$.
This is exactly the situation we are concerned with,
when we execute our max-strategy improvement algorithm.
For each $i \in \{ 1,\ldots, n \}$,
we now want to check whether or not 
$\rho(\vx_i) < \sem{e_i} \rho$.
If this is the case,
we moreover want 
to compute a max-strategy $\sigma_i'$ for $e_i$ such that 
$\rho(\vx_i) < \sem{\sigma_i'} \rho$.
Note that, 
since
$\rho(\vx_i) < \sem{e_i} \rho$,
we could also compute a max-strategy $\sigma_i'$ such that 
$\sem{\sigma_i'} \rho = \sem{e_i} \rho$.
If 
$\rho(\vx_i) < \sem{e_i} \rho$ 
does not hold,
then we set $\sigma_i' := \sigma_i$.
Finally,
the max-strategy 
$\sigma' := \{ \vx_1 = \sigma_1' ,\ldots, \vx_n = \sigma_n' \}$
is an improvement of $\sigma'$ w.r.t.\ $\rho$.

Given an equation $\vx = e$ and a variable assignment $\rho$, 
we must decide whether or not $\rho(\vx) < \sem{e}(\rho)$ holds,
and compute a max-strategy $\sigma'$ of $e$ such that 
$\rho(\vx) < \sem{\sigma'}(\rho)$ holds.
Recall that the semantic equations we are concerned with in this article are of  the form 
\begin{align}
  \vx = \max \left\{ e_1 , \ldots , e_k \right\}
\end{align} 

\noindent
where, for all $i \in \{1,\ldots,k\}$, each expression $e_i$ is either a constant or an expression of the form 
$\sem{s}^\sharp_{j\cdot}(\vx_1,\ldots,\vx_m)$.
Hence, 
we can answer the above question by answering the question 
for each argument $e_1,\ldots,e_k$ of the maximum separately.
It thus remains to find a method to 
check whether or not,
for a given statement $s$,
a given $j \in \{1,\ldots,m\}$, 
a given $c \in \R \cup \{\neginfty\}$, 
and a given $d \in \CR^m$,
$\sem{s}^\sharp_{j\cdot} (d) > c$ 
holds --- which is, by Lemma~\ref{l:diamand:is:np:complete}, a \complexclass{NP}-hard problem.
Our approach is to construct the following SAT modulo linear real arithmetic formula
(we use existential quantifiers 
 to improve readability):
\begin{align}
  \Psi(s,d,j,c)
  &:\equiv
  \exists \vv \in \R \;.\; \Psi(s,d,j) \wedge \vv > c
\\
  \Psi(s,d,j)
  &:\equiv
  \exists
  \vx \in \R^n,
  \vx' \in \R^n
  \;.\;
  T \vx \leq d \wedge \Psi(s) \wedge \vv = T_{j\cdot} \vx' %
\end{align}

\noindent
Here, $\Psi(s)$ is a formula that relates every $x \in \R^n$ with all elements from the set $\sem{s} \{ x \}$.
It is defined inductively over the structure of the statement $s$ as follows:
\begin{align}
  \Psi(s) 
    &:\equiv s
    && \text{if } s \text{ is a literal}
    \\
  \Psi(s_1 \wedge s_2)
    &:\equiv \Psi(s_1) \wedge \Psi(s_2)
    \\
  \Psi(s_1 \vee s_2) 
    &:\equiv 
      \left( \neg \va_{s_1\vee s_2}  \wedge \Psi(s_1) \right)
      \vee 
      \left( {\va_{s_1\vee s_2}} \wedge \Psi(s_2) \right)
\end{align}

\noindent
Here,
for every sub-formula $s_1 \vee s_2$ of $s$,
$\va_{s_1 \vee s_2}$ is a fresh Boolean variable.
The set of free variables of the formula 
$\Psi(s)$ is 
\begin{align}
  \{ \vx, \vx' \} 
  \cup 
  \{\va_{s_1\vee s_2} \mid s_1\vee s_2 \text{ is a sub-formula of } s  \}
  .
\end{align}

\noindent
The variables $\vx$ and $\vx'$ are $\R^n$-valued variables.
%
By construction,
$s[x/\vx ,\, x'/\vx']$ is satisfiable 
if and only if 
$\Psi(s)[x/\vx ,\, x'/\vx']$ is satisfiable
for all $x, x' \in \R^n$.
That is,
$s$ and $\Psi(s)$ 
are describing the same relation.
We therefore obtain the following lemma:

\begin{lem}
  \label{l:smt}
  $\sem{s}^\sharp_{j\cdot} (d) > c$ 
  if and only if
  $\Psi(s,d,j,c)$
  is satisfiable.
  \qed
\end{lem}

\noindent
The difference between the formula $s$ and the formula $\Psi(s)$ is that
the Boolean variables of the formula $\Psi(s)$ additionally describe a path through the formula.
More precisely,
a valuation for the variables from the set  
$\{\va_{s_1\vee s_2} \mid s_1\vee s_2 \text{ is a sub-formula of } s  \}$
describes a path through $s$.

Let $s$ be a statement, $d \in \CR^m$, $j \in \{1,\ldots,m\}$, 
and $c \in \R \cup \{\neginfty\}$.
Assume now that $\sem{s}^\sharp_{j\cdot} (d) > c$.
Our next goal is to compute a max-strategy $\sigma$ for the statement $s$ such 
that $\sem{\sigma}^\sharp_{j\cdot} (d) > c$.
By Lemma~\ref{l:smt},
there exists a model 
$M$ of $\Psi(s,d,j,c)$.
We define the max-strategy $\sigma_M$ for the statement $s$ 
recursively by
\begin{align}
  \sigma_M(s) 
    &:\equiv s
    && \text{if } s \text{ is a literal}
    \\
  \sigma_M(s_1 \wedge s_2)
    &:\equiv \sigma_M(s_1) \wedge \sigma_M(s_2)
    \\
  \sigma_M(s_1 \vee s_2) 
    &:\equiv 
    \begin{cases}
      \sigma_M(s_1) & \text{if } M(\va_{s_1 \vee s_2}) = 0 \\
      \sigma_M(s_2) & \text{if } M(\va_{s_1 \vee s_2}) = 1
    \end{cases}
  .
\end{align}

\noindent
By again applying Lemma~\ref{l:smt},
we get
$\sem{\sigma_M}^\sharp_{j\cdot} (d) > c$
and thus the following lemma:

\begin{lem}
  \label{l:smt:2}
  By solving the SAT modulo linear real arithmetic formula $\Psi(s,d,j,c)$ 
  that can be obtained from $s$ in linear time,
  we can decide, whether or not 
  $\sem{s}^\sharp_{j\cdot} (d) > c$ holds.
  From a model $M$ of this formula,
  we can, in linear time, obtain a $\vee$-strategy $\sigma_M$ for $s$ 
  such that $\sem{\sigma_M}^\sharp_{j\cdot} (d) > c$.
  \qed
\end{lem}


\begin{exa}
  \label{ex:running:smt:formula}
  We again continue 
  our running example,
  which is summarized in Figure~\ref{fig:run:ex}.
  We want to know,
  whether or not
  $\sem{s}^\sharp_{1\cdot} (0,0) \allowbreak > 0$
  holds.
  For that we compute a model $M$ of the 
  formula $\Psi(s,(0,0),1,0)$ which is given as follows:
  \begin{align}
    \Psi(s,(0,0),1,0)
     & \equiv
     \exists \vv \in \R \;.\; \Psi(s,(0,0)^\top,1) \wedge \vv > 0
   \\
   \Psi(s,(0,0),1)
   & \equiv
   \exists \vx \in \R^2, \vx' \in \R^2 \;.\; \vx_{1\cdot} \leq 0 \wedge -\vx_{1\cdot} \leq 0 \wedge \Psi(s) \wedge \vv = \vx'_{1\cdot}
   \\
   \Psi(s)
   & \equiv
   \Phi \wedge ((\neg \va_{\Phi_1 \vee \Phi_2} \wedge \Phi_1) \vee (\va_{\Phi_1 \vee \Phi_2} \wedge \Phi_2) )   
  \end{align}
  
  \noindent
  The formulas $\Phi, \Phi_1$, and $\Phi_2$ are defined in Figure~\ref{fig:run:ex}.
  $M = \{ a_{\Phi_1\vee\Phi_2} \mapsto 1 \}$ is a model,
  which gives us the max-strategy $\sigma_M \equiv \Phi \wedge \Phi_2$ for $s$.
  Thus, 
  by Lemma~\ref{l:smt:2},
  we have
  $\sem{\sigma_M}^\sharp_{1\cdot} (0,0) = \sem{\Phi \wedge \Phi_2}^\sharp_{1\cdot} (0,0) > 0$.
  \qed
\end{exa}

\noindent
We must still provide a method for computing the values for the Boolean variables of a model of  the formula $\Psi(s,d,j,c)$.
Most of the state-of-the-art SMT solvers, 
such as Yices \cite{yices,DBLP:conf/cav/DutertreM06},
support the computation of models directly;
the SMTLIB2 standard \cite{BarST-SMT-10} has a \texttt{get-assignment} command that can be used to extract the Boolean part of a model.
If this feature is not supported, 
one can compute the model,
or only the values for the Boolean variables,
using standard self-reduction techniques.


Recall that 
the semantic equations we are concerned with in this article are of  the form 
$
  \vx = \max \left\{ e_1 , \ldots , e_k \right\}
$,
where each expression $e_i$, for all $i \in \{1,\ldots,k\}$, 
is either a constant or an expression of the form 
$\sem{s}^\sharp_{j\cdot}(\vx_1,\ldots,\vx_m)$ where $s$ is a statement.
As discussed above,
we can check whether or not $\rho(\vx) < \sem{ \max \left\{ e_1 , \ldots , e_k \right\} }(\rho)$ holds, 
and if this is the case compute a max-strategy $\sigma'$ such that $\rho(\vx) < \sem{\sigma'}(\rho)$ holds,
by solving at most $k$ SAT modulo linear real arithmetic formulas, each of which can be constructed in linear time.
Equivalently, instead of running $k$ SMT queries, each obtaining a part of the next strategy, we can rename Boolean variables of these SMT formulas so that they are distinct and query the conjunction of the resulting formulas.

\begin{lem}
  \label{l:smt:3}
  Let $\vx = e$ be an abstract semantic equation,
  $\rho$ a variable assignment, 
  and $c \in \CR$.
  By solving a single SAT modulo linear real arithmetic formula 
  that can be obtained from $e$, $\rho$ and $c$ in linear time,
  we can decide, whether or not $\sem{e} \rho > c$ holds.
  From a model $M$ of this formula,
  provided that $\sem{e} \rho > c$ holds,
  we can in linear time obtain a max-strategy $\sigma_M$ for $e$ 
  such that $\sem{\sigma_M} \rho > c$.
  \qed
\end{lem}

Remark that we did not discuss how to choose the next max-strategy $\sigma'$, except that it should satisfy $\rho(\vx) < \sem{\sigma'}(\rho)$ (which is ensured by the SMT-solving step).
Indeed, there could be many different suitable $\sigma'$s, and the SMT-solver may return any of them.
There is however at least one that is \emph{locally optimal}, that is, $\sem{\sigma'}(\rho)$ is maximal, otherwise said $\sem{\sigma'}(\rho) = \sem{e}(\rho)$.
Future work should include experiments on the performance impact of using the locally optimal strategies 
instead of just any strategies.
 
It is possible to obtain a locally optimal strategy by repeated calls to the SMT-solvers.
A naive method would be to query the SMT-solver for a $\sigma''$ such that $\sem{\sigma'}(\rho) < \sem{\sigma''}(\rho)$, 
then for a $\sigma'''$ such that $\sem{\sigma''}(\rho) < \sem{\sigma'''}(\rho)$ and so on until there is no locally better strategy; the last strategy obtained is thus locally optimal.
A less naive method would be to take a rough bound $M \geq \sem{e}(\rho)$ and perform binary search in the interval $[\sem{\sigma'}(\rho), M]$: at each step, maintain an interval $[a,b]$ and query whether there exists $\sigma''$ such that $\sem{\sigma'}(\rho) \geq \frac{a+b}{2}$; if so, replace $a$ by $\frac{a+b}{2}$ and restart, if not, replace $b$ by $\frac{a+b}{2}$ and restart. 
The SMT-solving community is now considering the problem of \emph{optimization modulo theory} \cite{Sebastiani_Tomasi_IJCAR12} and we can hope for progress in this respect.


%% file: lower_bound.tex
\subsection{A Lower Bound on the Complexity}
\label{s:lower}

In this section we show that the problem of computing abstract semantics of programs
w.r.t.\ the interval domain is \piptwo-hard.
\piptwo-hard problems are conjectured to be harder than both \complexclass{NP}-complete 
and \complexclass{coNP}-complete problems.
For further information regarding the polynomial-time hierarchy see, for instance,
\citet{Stockmeyer76,Papadimitriou94}.

\begin{thm}
  \label{t:lower}
  The problem of deciding,
  whether,
  for a given program $G$, 
  a given template constraint matrix $T$,
  and a given program point $v$,
  $\Values^\sharp[v] > \neginftyvar$ holds,
  is \piptwo-hard.

  The problem remains \piptwo-hard even if the program variables are abstracted at a single program point and the 
  template constraint matrix $T$ is restricted to a single variable $x$ and a single constraint of the form $x \leq B$.
\end{thm}

\begin{proof}
  We reduce the \piptwo-complete problem of
  deciding the truth of a $\forall^*\exists^*$ propositional formula
  \citep{DBLP:journals/tcs/Wrathall76}
  to our static analysis problem. 
  Let 
  \begin{align}
    \Phi \equiv \forall \vx_1 , \ldots  , \vx_n . \exists \vy_1, \ldots, \vy_m \,.\, \Phi'
  \end{align}
  
  \noindent
  be a formula without free variables, 
  where $\Phi'$ is a propositional formula.
We consider the analysis of the following pseudo-C program, where $n$ is a constant:
\begin{lstlisting}
x = 0;
while (x < $2^n$) {
  z = x;
  if (x >= $2^{n-1}$) { $x_n$=1; x -= $2^{n-1}$; } else { $x_n$=0; }
  $\vdots$
  if (x >= $2^{1-1}$) { $x_1$=1; x -= $2^{1-1}$; } else { $x_1$=0; }
  choose($y_1,\dots,y_m$);
  if ($\Phi'(x_1,\dots,x_n,y_1,\dots,y_m)$) {
    x++;
  }
}
\end{lstlisting}

\noindent
  In intuitive terms: this program initializes the program variable \lstinline|x| to $0$. 
  Then, it enters a loop: compute into $x_1,\dots,x_n$ 
  the binary decomposition of \lstinline|x|, and non-deterministically choose 
  $y_1,\dots,y_m$. 
  If  $\Phi'$ is true, it increments \lstinline|x| by one and loops, unless \lstinline|x| reaches $2^n$ in which case it terminates;
  otherwise, it just loops.
  Thus, there exists a terminating computation if and only if $\Phi$ holds.

  We reformulate the above pseudo-C program into the program
   $G = (N,E,\start)$
  that uses only one program variable $\vx$,  
  where
  \begin{enumerate}[(1)]
    \item
      $N = \{ \start, 1, 2 \}$ is the set of program points, 
      and
    \item
  $E = \{ (\start, \vx' = 0, 1), \allowbreak (1,s,1), \allowbreak (1,\vx \geq 2^{n},2) \}$
  is the set of control-flow edges,
  where
  \begin{align*}
    s
    \quad\equiv\quad &
    \vz_n = \vx
    \\
    &\quad\wedge
      \left( 
        ( \vz_n \geq 2^{n-1} \wedge \vz_{n-1} = \vx - 2^{n-1} \wedge \vx_n = 1 )
        \vee
        ( \vz_n \leq 2^{n - 1} - 1 \wedge \vx_n = 0 )
      \right)
    \\&\quad\wedge
      \cdots \\
    &\quad\wedge
      \left( 
        ( \vz_1 \geq 2^{1-1} \wedge \vz_0 = \vx_1 - 2^{1-1} \wedge \vx_1 = 1 )
        \mid 
        ( \vz_1 \leq 2^{1-1} - 1 \wedge \vx_1 = 0 )
      \right)
    \\&\quad\wedge
      s(\Phi') 
    \\&\quad\wedge
      \vx' = \vx + 1 
    .
  \end{align*}

  \noindent
  The statement $s(\Phi')$ is obtained by taking formula $\Phi'$ in negation normal form (all negations pushed to the leaves), leaving the Boolean structure in place and replacing each positive literal $x$ by a test $x=1$ and each negative literal $\neg x$ by a test $x=0$.
  \end{enumerate}

  \noindent
  With this formalization, $\Phi$ 
  holds if and only if
  $\Values[2] \neq \emptyset$.
  For the abstraction,
  we consider the interval domain, or even simply the domain
  of upper bounds on $\vx$ (i.e., we have constraints of the form $\vx \leq b$).
  By considering the Kleene iteration,
  it is easy to see that 
  $\Values[2] \neq \emptyset$ holds if and only if  
  $\Values^\sharp[2] > \neginftyvar$ holds.
  Thus $\Phi$ holds if and only if $\Values^\sharp[2] > \neginftyvar$ holds.
\end{proof}



%% file: exp_example.tex
\subsection{An Example with Exponential Running Time Behavior}


Recall that the number of strategy improvement steps is exponentially bounded by the size of the input.
Each step consists in one phase of SMT-solving for linear real arithmetic 
followed by solving a linear program of polynomial size.
Thus, each step can be performed in exponential time.
Therefore,
the whole algorithm can be executed in exponential time. 

We shall now see that our algorithm 
takes exponential time on the instances 
that are similar to the instances generated from the reduction 
in the proof of Theorem~\ref{t:lower}.
The instances generated from the reduction require $\Theta(2^n)$ steps.
However, the input is of size $\mathcal O(n^2)$,
because the numbers $2^{n-1}, 2^{n-2}, \ldots, 2^0$ require space $\Theta (n^2)$.
We modify the instances generated from the reduction in such a way that the sizes of the programs are in $\mathcal O(n)$.
We achieve this by introducing auxiliary variables for the numbers $2^{n-1}, 2^{n-2}, \ldots, 2^0$.
For all $n \in \N$,
we define the program $G_n = (N,E,\start)$,
where
\begin{align}  
  N &= \{ \start, 1 \} , \\
  E &= \{ (\start, \vx_1' = 0, 1), (1,s,1) \}
  \text{,}
\end{align}

\noindent
with
\begin{align}
    s
    \;\equiv\;&\; \vy_1 = 1 \wedge \vy_2 = 2 \vy_1 \wedge \cdots \wedge \vy_n = 2 \vy_{n-1} \wedge \vz_n = \vx_1 \\
    \wedge &\;( \vz_n \geq \vy_n \wedge \vz_{n-1} = \vz_n - \vy_n \vee \vz_n \leq \vy_n - 1 \wedge \vz_{n-1} = \vz_n ) \\
    \wedge &\; \cdots \\
    \wedge &\;( \vz_1 \geq \vy_1 \wedge \vz_{0} = \vz_1 - \vy_1 \vee \vz_1 \leq \vy_1 - 1 \wedge \vz_{0} = \vz_1 ) \\
    \wedge &\;\vx_1' = \vx_1 + 1 
    .
\end{align}
  
\noindent
Here, $\vx_1$ is the only program variable.
It is sufficient to use the template constraint matrix $T = \begin{pmatrix} 1 \end{pmatrix}$, 
which corresponds to the  template $\vx_1$.
That is, we are only interested in the upper bound on the value of the variable $\vx_1$.
Remark that the strategy iteration does not depend on the
strategy improvement operator in use:
at any time there is exactly one possible improvement, until the least solution is reached.
All strategies for the statement $s$ will be encountered.
Thus, the strategy improvement algorithm performs $2^n$ strategy improvement steps.
Since the size 
of $G_n$ is $\Theta(n)$,
exponentially many strategy improvement steps are performed.

%% file: upper.tex
\subsection{An Upper Bound on the Complexity}
\label{s:upper}

In Section~\ref{s:lower},
we have provided a lower bound on the complexity of computing abstract semantics 
w.r.t. the template linear domains. 
The associated decision problem is not only 
\piptwo-hard, but in fact $\Pi^p_2$-complete:

\begin{thm}
\label{t:upper}
  The problem of deciding,
  whether, 
  for a given program $G$, 
  a given template constraint matrix $T$,
  and a given program point $v$,
  $\Values^\sharp[v] > \neginftyvar$ holds,
  is in \piptwo.
\end{thm}

\begin{proof}
  We consider the negation of the above problem:
  for a given program $G$,
  a given template constraint matrix $T$,
  a given program point $v$,
  and a given $i \in \{1,\ldots,m \}$,
  decide whether
  $\Values^\sharp_{i\cdot}[v] = \neginfty$;
  we shall now show that this problem is in $\Sigma^p_2$.

  In non-deterministic polynomial time we can guess a max-strategy $\sigma$ for 
  $\E' := \E(G,T)$ and a set $\vX^\infty$ of variables that have the value $\infty$; these will form the witness for the initial existential quantifier.
  We can evaluate the max-strategy $\sigma$ w.r.t.\ the set
  of variables $\vX^\infty$ assigned to $+\infty$ in polynomial 
  time using linear programming (cf. Subsection~\ref{ss:veal:max:strat}).
  Let $\rho_{\sigma,\vX^\infty}$ denote the resulting variable assignment.

  We shall now show that checking whether this strategy (and set of infinite variables) is stable is in
  \complexclass{co-NP}.
  Because of Lemma~\ref{l:diamand:is:np:complete},
  we can use an \complexclass{NP} oracle to check 
  whether there exists an improvement of the strategy $\sigma$ w.r.t.\ $\rho_{\sigma,\vX^\infty}$, which is exactly the negation of being stable.

  If the strategy is stable, we know that 
  $\rho_{\sigma,\vX^\infty} \geq \mu\sem{\E'}$ holds.
  Therefore,
  by Lemma~\ref{l:eqs:abs_sem},
  we have $\rho_{\sigma,\vX^\infty}(\vx_{v,i}) \geq \Values^\sharp_{i\cdot}[v]$
  for all program points $v \in N$ and all $i \in \{1,\ldots,m\}$.
  Since we also know that there exists some max-strategy $\sigma$ and some set $\vX^\sigma$ such that 
  $\rho_{\sigma,\vX^\infty} = \mu\sem{\E'}$,
  we accept,
  whenever $\rho_{\sigma,\vX^\infty}(\vx_{v,i}) = \neginfty$ holds.
\end{proof}



%% file: experiments.tex
\section{Experimental Results}
\label{a:exp:res}

We have implemented our presented max-strategy improvement algorithm;
our prototype should however be considered as a proof-of-concept.
Benchmark results for real examples are left for future work.

The algorithm is implemented in OCaml 3.10.2;
it uses Yices 1.0.27 \cite{yices,DBLP:conf/cav/DutertreM06} 
for computing models for SAT modulo linear real arithmetic formulas;
for solving the occurring linear programming problems it uses QSOpt-Exact 2.5.6 \cite{DBLP:journals/orl/ApplegateCDE07,espinoza2006},
an exact arithmetic version of QSOpt.
We made our experiments under Debian Linux (Lenny) running under Parallels Desktop 4 on an Apple MacBook 
(2.16 GHz Intel Core 2 Duo, 2GB 667 MHz DDR2 SDRAM).
\lstdefinelanguage{prg}{otherkeywords={->,:=,|,;,>=,<=,:},keywords={->},keywordstyle=\textbf}
\lstset{language=prg}
Our solver takes as input a text file that contains the program 
and the linear templates to be used for the analysis.
%
%
%
%
%
\begin{figure}
\begin{center}
\begin{tabular}{|r|r|r|r|r|r|}
\hline
\multicolumn{1}{|c|}{$n$} & 
\multicolumn{1}{|c|}{user} & 
\multicolumn{1}{|c|}{number} & 
\multicolumn{1}{|c|}{number} & 
\multicolumn{1}{|c|}{number} 
\\ 
\multicolumn{1}{|c|}{} & 
\multicolumn{1}{|c|}{time} & 
\multicolumn{1}{|c|}{of} & 
\multicolumn{1}{|c|}{of} & 
\multicolumn{1}{|c|}{of} 
\\ 
\multicolumn{1}{|c|}{} & 
\multicolumn{1}{|c|}{(sec)} & 
\multicolumn{1}{|c|}{improvement} & 
\multicolumn{1}{|c|}{SMT} & 
\multicolumn{1}{|c|}{linear} 
\\ 
\multicolumn{1}{|c|}{} & 
\multicolumn{1}{|c|}{} & 
\multicolumn{1}{|c|}{steps} & 
\multicolumn{1}{|c|}{queries} & 
\multicolumn{1}{|c|}{programs} 
\\ 
\hline
\hline
1 & 0.10 & 5 & 14 & 8 \\
\hline
2 & 0.17 & 7 & 34 & 12 \\
\hline
3 & 0.38 & 11 & 76 & 20 \\
\hline
4 & 1.02 & 19 & 170 & 36 \\
\hline
5 & 3.64 & 35 & 384 & 68 \\
\hline
6 & 6.97 & 67 & 870 & 132 \\
\hline
7 & 26.02 & 131 & 1964 & 260 \\
\hline
8 & 31.53 & 259 & 4402 & 516 \\
\hline
9 & 95.22 & 515 & 9784 & 1028 \\
\hline
10 & 207.62 & 1027 & 21566 & 2052 \\
\hline
\end{tabular}
\end{center}
\caption{Benchmark for the prototypical implementation}
\label{fig:bench}
\end{figure}
The benchmark results for the example of Section~\ref{s:upper} are shown in Figure \ref{fig:bench}.
The number of max-strategy improvement steps grows --- as expected --- exponentially in $n$.
Briefly, the implementation solves $2$ linear programming problems and at most $2 ( 2 n + 1) = 4n + 2$ SMT queries per max-strategy improvement step.
The factor $2$ comes from the fact that we have $2$ program points and the factor $(2n + 1)$ from the fact that we have $(2n + 1)$ templates.
We emphasize that the example is created artificially.
Since the problem we are solving is \piptwo-complete,
it is not surprising that there exists an example that does not scale.

For the running example of this article (Example~\ref{ex:running:0}),
our solver computes the correct result after $0.05$ seconds.

There are also many possibilities for improving the implementation.
On the limited number of examples that we tried with our proof-of-concept implementation, 
the main computational expense comes from the linear programs that have to be solved.
This is mainly due to the fact that we use an exact arithmetic simplex solver and we solve every occurring linear program from scratch
although we know beforehand that the linear problems that we have to solve are feasible.
Instead of solving each linear program from scratch, 
one could use the information obtained from the previously solved 
linear programs (that are similar).
One can also utilize 
the information obtained from the SMT solver in order to obtain a feasible basis to start the simplex method with.

%% file: conclusion.tex
\section{Conclusion and further research directions}
\label{s:conclusion}
We have proposed a method for computing the least fixpoints in template linear constraint domains
(e.g., Cartesian products of intervals) 
of transition systems specified using linear real arithmetic formulas. This allows finding the strongest invariant in this domain of a loop consisting only in linear assignments and non-strict linear inequalities over the real numbers.

Because it distinguishes individual paths in the program, our method does not suffer from the imprecision induced by convex hull operations. These paths are looked up on demand, as results from satisfiability testing, therefore avoiding memory blowup. Our technique, however, has exponential worst case complexity, which is hardly surprising since the decision problem associated with our computation is \piptwo-complete. Due to limited resources, we have so far not been able to implement it in a tool capable of running on real examples.

It is quite obvious that, due to the use of SMT queries, the size of the problems given as input, and their branching structure, must be limited. One method for limiting the size of the SMT formulas is to decompose the program into statements, thus adding more points at which states are abstracted, as proposed by \cite{Monniaux_Gonnord_SAS11}: this simplifies the problem, but may reduce precision; another method is to restrict the analysis to a subset of the variables, determined by some form of dependency analysis.

The restriction to linear templates and linear statements may seem onerous.
It might be possible to apply the same ideas for non-linear templates \citep{gawlitza_sas_10}.
With respect to non-linear statements, a possibility is to linearize them \citep{Mine_PhD,DBLP:conf/vmcai/Mine06}: for short, assuming $A \leq x \leq B$ where $A$ and $B$ are constants, then the nonlinear constraint $z = x y$ may be abstracted by the linear constraint $(A y \leq x y \leq B y \land y \geq 0) \lor (B y \leq x y \leq A y \land y < 0)$. 
If the assumptions made by the linearization are found not to hold for the fixed point computed by the max-strategy iteration technique, one has to relax these assumptions and restart the solving process.

More generally, one may envision a nesting of two iteration schemes: the inner scheme solving exactly, using 
max-strategy iteration, a simplification of the concrete program, the outer scheme iterating over possible simplifications.
The outer scheme would deal with all program features not supported by our max-strategy iteration algorithm.
Consider pointers, for instance: the outer scheme could temporarily assume that $x$ and $y$ may be aliased, while $z$ is not aliased with anything, and then rewrite the program according to these assumptions in order to obtain a pointer-free program (may-alias information becomes non-deterministic choice, while must-aliased variables are merged).
This outer iteration may be ascending and optimistic, starting with strong assumptions on the program and relaxing them progressively as the results of the inner scheme invalidate them, or decreasing and pessimistic, starting with weak assumptions and strengthening them progressively as the results of the inner scheme show them to be too severe. Such mixed approaches would cope with programs features not directly supported by our max-strategy iteration solver. Further work is needed in this direction to ascertain which techniques are usable.

Another problem is finding suitable templates --- while there exist obvious choices in some cases (intervals for getting rough invariants of control applications, difference bounds for scheduling applications, etc.), there is no generic method for obtaining good templates. \citet{DBLP:conf/sas/AmatoPS10} proposed finding templates using principal component analysis, but it is yet unclear whether this approach suited to practical problems. A simple solution may be to run some conventional polyhedral analysis, and keeping the directions of the polyhedra obtained before widening.

Our max-strategy iteration algorithms only deal with real numerical values.
We can cope with integers by relaxing them to reals, with the usual precautions ($x < y$ converted to $x \leq y-1$). Another possible extension is to integrate Boolean types, or more generally finitely enumerated types, into the invariant, or equivalently, to insert them implicitly into the control flow.

An intriguing extension of our framework is the case where the control flow is specified implicitly. The problem considered in this article is expressed as a control-flow graph given by a list of nodes and statements over the transitions. Now consider the addition of $n$ Boolean variables to the system: a common method to encode such variables in a transition system is to distinguish all Boolean combinations and every control node, and thus multiply the number of control nodes by~$2^n$. Clearly, we would prefer to work directly on the transition relation of the original program, which would include free Boolean variables encoding the departure and arrival control states, and consider our abstract reachability problem on programs expressed using this succinct representation.
Since this problem includes Boolean reachability (also known as the reachability problem for succinctly represented graphs), which is \complexclass{PSPACE}-complete~\cite{Papadimitriou_Yannakakis_86}, it is \complexclass{PSPACE}-hard. 
Our strategy iteration approach can be extended to show that it is in \complexclass{coNEXPTIME}. We conjecture that it is  \complexclass{coNEXPTIME}-complete, but we have so far not been able to prove it. It is also unknown whether some practically useful algorithms, perhaps based on binary decision diagrams (BDDs), could be devised for this problem.
